\documentclass{article}

\usepackage{hyperref}

\usepackage{amssymb,amsmath,xspace,enumerate,gensymb}
\usepackage{amsthm}

\usepackage[pagewise]{lineno}
\usepackage{graphicx}
\usepackage{xspace}
\usepackage{arydshln}
\usepackage{rotating}
\usepackage{listings}
\usepackage{soul}
\usepackage{ifmtarg}
\usepackage{stmaryrd} \usepackage{paralist}
\usepackage{chngcntr}
\usepackage{authblk}

\usepackage{tikz}
\usetikzlibrary{arrows,decorations.pathmorphing,backgrounds,calc,positioning,fit,petri,matrix,backgrounds, decorations.pathreplacing}
\usetikzlibrary{shapes.geometric}
\usepackage{pgffor}
\usepackage{extarrows}
\usepackage{appendix}

\newif\ifcomments
\newif\ifchanges
\commentsfalse\changesfalse
\commentstrue\changestrue

\begin{document}

\newcommand  {\myclass} [1]  {\ensuremath{\textsc{#1}}}

\newcommand{\StaClass}[1]{\myclass{#1}\xspace}

\newcommand{\DynClass}[1]{\myclass{Dyn#1}\xspace}
\newcommand{\dDynClass}[1]{\myclass{$\Delta$-Dyn#1}\xspace}

\newcommand  {\myproblem} [1] {\textsc{#1}}

\newcommand{\problemIndent}{\hspace{5mm}}
\newcommand  {\problemdescr} [3] {
    \vspace{3mm}
    \def\Name{#1}
    \def\Input{#2}
    \def\Question{#3}
     \problemIndent\begin{tabular}{r p{\columnWidth}r}      \textit{Problem:} & \myproblem{\Name} \\
      \textit{Input:} & \Input \\
      \textit{Question:} & \Question
     \end{tabular}
    \vspace{3mm}
    }

\newcommand  {\querydescr} [3] {
\vspace{3mm}
\def\Name{#1}
\def\Input{#2}
\def\Question{#3}
  \problemIndent\begin{tabular}{r p{\columnWidth}r}  \textit{Query:} & \myproblem{\Name} \\
  \textit{Input:} & \Input \\
  \textit{Question:} & \Question
  \end{tabular}
\vspace{3mm}
}

\newcommand  {\dynproblemdescr} [4] {
    \vspace{3mm}
    \def\Name{#1}
    \def\Input{#2}
    \def\Updates{#3}  
    \def\Question{#4}
    \problemIndent\begin{tabular}{r p{\columnWidth}r}      \textit{Query:} & \myproblem{\Name} \\
      \textit{Input:} & \Input \\
      \textit{Question:} & \Question
    \end{tabular}
    \vspace{3mm}
}
\newcommand{\dynProbDescr}[4]{\dynproblemdescr{#1}{#2}{#3}{#4}}

\newcommand  {\problem}[1] {\myproblem{#1}}

\newcommand{\dynProb}[1] {\myproblem{Dyn(#1)}}

\newcommand{\class}{\calC}

\newcommand  {\TIME}    {\myclass{TIME}}
\newcommand  {\DTIME}   {\myclass{DTIME}}
\newcommand  {\NTIME}   {\myclass{NTIME}}
\newcommand  {\ATIME}   {\myclass{ATIME}}
\newcommand  {\SPACE}   {\myclass{SPACE}}
\newcommand  {\DSPACE}   {\myclass{DSPACE}}
\newcommand  {\NSPACE}  {\myclass{NSPACE}}
\newcommand  {\coNSPACE}        {\myclass{coNSPACE}}

\newcommand     {\LOGCFL}     {\myclass{LOGCFL}}
\newcommand     {\LOGDCFL}     {\myclass{LOGDCFL}}
\newcommand     {\LOGSPACE}     {\myclass{LOGSPACE}}
\newcommand     {\NLOGSPACE}     {\myclass{NLOGSPACE}}
\newcommand     {\classL}   {\myclass{L}}
\newcommand     {\NL}   {\myclass{NL}}
\newcommand     {\coNL}   {\myclass{coNL}}
\renewcommand   {\P}    {\myclass{P}}
\newcommand     {\myP}    {\myclass{P}}
\newcommand     {\PTIME}    {\myclass{PTIME}}
\newcommand     {\NP}   {\myclass{NP}}
\newcommand     {\NPC}   {\myclass{NPC}}
\newcommand     {\PH}   {\myclass{PH}}
\newcommand     {\coNP} {\myclass{coNP}}
\newcommand     {\NPSPACE}      {\myclass{NPSPACE}}
\newcommand     {\PSPACE}       {\myclass{PSPACE}}
\newcommand     {\IP}   {\myclass{IP}}
\newcommand     {\POLYLOGSPACE} {\myclass{POLYLOGSPACE}}
\newcommand     {\DET}  {\myclass{DET}}
\newcommand     {\EXP}  {\myclass{EXP}}
\newcommand     {\NEXP}  {\myclass{NEXP}}
\newcommand     {\EXPTIME}  {\myclass{EXPTIME}}
\newcommand     {\TWOEXPTIME}  {\myclass{2-EXPTIME}}
\newcommand     {\TWOEXP}  {\myclass{2-EXP}}
\newcommand     {\NEXPTIME}  {\myclass{NEXPTIME}}
\newcommand     {\coNEXPTIME}  {\myclass{coNEXPTIME}}
\newcommand     {\EXPSPACE}  {\myclass{EXPSPACE}}
\newcommand     {\RP}   {\myclass{RP}}
\newcommand     {\RL}   {\myclass{RL}}
\newcommand     {\coRP} {\myclass{coRP}}
\newcommand     {\ZPP}  {\myclass{ZPP}}
\newcommand     {\BPP}  {\myclass{BPP}}
\newcommand     {\PP}   {\myclass{PP}}
\newcommand     {\NC}   {\myclass{NC}}
\newcommand     {\myAC}   {\myclass{AC}}
\newcommand     {\myACZ}   {\myAC^0}
\newcommand     {\myACZP}   {\myAC^0+\text{\myclass{Mod 2}}}
\newcommand     {\myTC}   {\myclass{TC}}
\newcommand     {\myTCZ}   {\myTC^0}

\newcommand     {\SAC}   {\myclass{SAC}}
\newcommand     {\ACC}   {\myclass{ACC}}
\newcommand     {\tc}   {\myclass{TC}}   \newcommand     {\PPoly}{\myclass{\mbox{P}/\mbox{Poly}}} 
\newcommand     {\FOarb}   {\myclass{FO(arb)}}

\newcommand     {\NLIN}   {\myclass{NLIN}}
\newcommand     {\DLIN}   {\myclass{DLIN}}

\newcommand  {\APTIME}   {\myclass{APTIME}}
\newcommand  {\ALOGSPACE}   {\myclass{ALOGSPACE}}

\newcommand{\FO}{\StaClass{FO}}
\newcommand{\MSO}[1][\quant]{\StaClass{MSO}}
\newcommand{\EMSO}{\StaClass{$\exists$MSO}}
\newcommand{\QFO}[1][\quant]{\StaClass{\ensuremath{#1}FO}}
\newcommand{\cQFO}[1][\quant]{\StaClass{\ensuremath{\overline{#1}}FO}}
\newcommand{\EFO}{\QFO[\exists^*]}
\newcommand{\AFO}{\QFO[\forall^*]}
\newcommand{\AEFO}{\StaClass{$\forall/\exists$FO}}
\newcommand{\CQ}[1][]{\StaClass{CQ}}
\newcommand{\UCQ}[1][]{\StaClass{UCQ}}
\newcommand{\CQneg}[1][]{\StaClass{CQ\ensuremath{^{\mneg}}}}
\newcommand{\UCQneg}[1][]{\StaClass{UCQ\ensuremath{^{\mneg}}}}
\newcommand{\Prop}{\StaClass{Prop}}
\newcommand{\QF}{\StaClass{QF}}
\newcommand{\PropCQ}{\StaClass{PropCQ}}
\newcommand{\PropUCQ}{\StaClass{PropUCQ}}
\newcommand{\PropCQneg}{\StaClass{PropCQ{\ensuremath{^{\mneg}}}}}
\newcommand{\PropUCQneg}{\StaClass{PropUCQ{\ensuremath{^{\mneg}}}}}

\newcommand{\mneg}{\neg} 
\newcommand{\DynTC}{\DynClass{TC}}

\newcommand{\DynProp}{\DynClass{Prop}}
\newcommand{\DynPropbi}{\DynClass{Prop}^*\xspace}
\newcommand{\DynProj}{\DynClass{Projections}}
\newcommand{\DynQF}{\DynClass{QF}}
\newcommand{\DynQFbi}{\DynClass{QF}^*\xspace}
\newcommand{\DynFO}{\DynClass{FO}}
\newcommand{\DynACZ}{\DynClass{AC$^0$}}
\newcommand{\DynACZP}{\DynClass{AC$^0 + $ Mod 2}}
\newcommand{\DynFOpos}{\DynClass{FO$^{\wedge, \vee}$}}
\newcommand{\DynFOand}{\DynClass{FO$^{\wedge}$}}

\newcommand{\DynC}{\DynClass{$\class$}}
\newcommand{\DynUCQ}{\DynClass{UCQ}}
\newcommand{\DynCQ}{\DynClass{CQ}}
\newcommand{\DynUCQneg}{\DynClass{UCQ$^\mneg$}}
\newcommand{\DynCQneg}{\DynClass{CQ$^\mneg$}}
\newcommand{\DynCQPM}{\DynCQneg}
\newcommand{\DyncQFO}{\DynClass{$\cquant$FO}}

\newcommand{\DynQFO}[1][\quant]{\DynClass{\QFO[#1]}}
\newcommand{\DynEFO}{\DynQFO[\exists^*]}
\newcommand{\DynAFO}{\DynQFO[\forall^*]}

\newcommand{\DynAEFO}{\DynClass{$\forall/\exists$FO}}
\newcommand{\DynAND}{\DynClass{PropCQ}}
\newcommand{\DynAnd}{\DynAND}
\newcommand{\DynPropCQ}{\DynAND}

\newcommand{\DynPropPos}{\DynClass{PropUCQ}}
\newcommand{\DynPropAO}{\DynPropPos}
\newcommand{\DynPropUCQ}{\DynPropPos}

\newcommand{\DynAndNeg}{\DynClass{PropCQ{\ensuremath{^{\mneg}}}}}
\newcommand{\DynPropCQneg}{\DynAndNeg}
\newcommand{\DynPropUCQneg}{\DynClass{PropUCQ{\ensuremath{^{\mneg}}}}}

\newcommand{\DynOrNeg}{\DynClass{Or{\ensuremath{^{\mneg}}}}}

\newcommand{\dDynProp}{\dDynClass{Prop}}
\newcommand{\dDynPropPos}{\dDynClass{PropUCQ}}
\newcommand{\dDynAndOr}{\dDynPropPos}
\newcommand{\dDynQF}{\dDynClass{QF}}
\newcommand{\dDynFO}{\dDynClass{FO}}
\newcommand{\dDynFOpos}{\dDynClass{FO$^{\wedge, \vee}$}}
\newcommand{\dDynFOand}{\dDynClass{FO$^{\wedge}$}}
\newcommand{\dDynC}{\dDynClass{$\class$}}
\newcommand{\dDynUCQ}{\dDynClass{UCQ}}
\newcommand{\dDynCQ}{\dDynClass{CQ}}
\newcommand{\dDynUCQneg}{\dDynClass{UCQ$^\mneg$}}
\newcommand{\dDynCQneg}{\dDynClass{CQ$^\mneg$}}
\newcommand{\dDynCQPM}{\dDynCQneg}

\newcommand{\dDynQFO}[1][\quant]{\dDynClass{\QFO[#1]}}
\newcommand{\dDynEFO}{\dDynQFO[\exists^*]}
\newcommand{\dDynAFO}{\dDynQFO[\forall^*]}

\newcommand{\dDynAEFO}{\dDynClass{$\forall/\exists$FO}}
\newcommand{\dDynAND}{\dDynPropCQ}
\newcommand{\dDynAnd}{\dDynAND}
\newcommand{\dDynConj}{\dDynClass{Conj}}
\newcommand{\dDynPropAO}{\dDynClass{Prop$^{\wedge, \vee}$}}
\newcommand{\dDyncQFO}{\dDynClass{$\cquant$FO}}
\newcommand{\dDynPropUCQneg}{\dDynClass{PropUCQ{\ensuremath{^{\mneg}}}}}
\newcommand{\dDynPropUCQ}{\dDynClass{PropUCQ}}
\newcommand{\dDynPropCQneg}{\dDynClass{PropCQ{\ensuremath{^{\mneg}}}}}
\newcommand{\dDynPropCQ}{\dDynClass{PropCQ}}

\newcommand{\equalcardinality}{\textsc{EqualCardinality}\xspace}\newcommand{\reach}{\textsc{Reach}\xspace}\newcommand{\altreach}{\textsc{Alt-Reach}\xspace}
\newcommand{\stgraph}{$s$-$t$-graph\xspace}
\newcommand{\stgraphs}{$s$-$t$-graphs\xspace}
\newcommand{\reachQ}{\textsc{Reach}\xspace}
\newcommand{\streachQ}{\textsc{$s$-$t$-Reach}\xspace}
\newcommand{\streachabilityquery}{$s$-$t$-reachability query\xspace}
\newcommand{\stTwoPath}{\problem{$s$-$t$-Two\-Path}\xspace}
\newcommand{\sTwoPath}{\problem{$s$-Two\-Path}\xspace}
\newcommand{\clique}[1]{\problem{$#1$-Clique}\xspace}
\newcommand{\colorability}[1]{\problem{$#1$-Col}\xspace}
\newcommand{\streach}{$s$-$t$-Reach}
\newcommand{\streachp}{\problem{\streach}\xspace}
\newcommand{\layeredstreach}[1]{#1-Layered-$s$-$t$-Reach}
\newcommand{\layeredstreachp}[1]{\problem{\layeredstreach{#1}}\xspace}

\newcommand{\dynClique}[1]{\clique{#1}}
\newcommand{\dynColorability}[1]{\colorability{#1}}

\newcommand{\probEqualCardinalityText}{EqualCardinality}
\newcommand{\EqualCardinality}{\problem{\probEqualCardinalityText}\xspace}
\newcommand{\EqualCardinalityDescr}{\problemdescr{\probEqualCardinalityText}{Unary relations $A$ and $B$}{Do $A$ and $B$ have the same cardinality?\xspace}}

\newcommand{\dynEqualCardinality}{\dynProb{\probEqualCardinalityText}\xspace}
\newcommand{\dynEqualCardinalityDescr}{\dynProbDescr{\probEqualCardinalityText}{Unary relations $A$ and $B$}{Element insertions and deletions}{Do $A$ and $B$ have the same cardinality?\xspace}}

\newcommand{\dynReachQ}{\dynProb{\textsc{Reach}}\xspace}
\newcommand{\dynstReachQ}{\dynProb{\textsc{$s$-$t$-Reach}}\xspace}

\newcommand{\dynstTwoPath}{\dynProb{\stTwoPath}\xspace}
\newcommand{\dynsTwoPath}{\dynProb{\sTwoPath}\xspace}

\newcommand{\dynlayeredstreach}[1]{Dyn-#1-Layered-$s$-$t$-Reach}
\newcommand{\dynlayeredstreachp}[1]{\problem{\dynlayeredstreach{#1}}\xspace}

\newcommand{\mtext}[1]{\textsc{#1}}

\providecommand {\calA}      {{\mathcal A}\xspace}
\providecommand {\calB}      {{\mathcal B}\xspace}
\providecommand {\calC}      {{\mathcal C}\xspace}
\providecommand {\calD}      {{\mathcal D}\xspace}
\providecommand {\calE}      {{\mathcal E}\xspace}
\providecommand {\calF}      {{\mathcal F}\xspace}
\providecommand {\calG}      {{\mathcal G}\xspace}
\providecommand {\calH}      {{\mathcal H}\xspace}
\providecommand {\calK}      {{\mathcal K}\xspace}
\providecommand {\calI}      {{\mathcal I}\xspace}
\providecommand {\calL}      {{\mathcal L}\xspace}
\providecommand {\calM}      {{\mathcal M}\xspace}
\providecommand {\calN}      {{\mathcal N}\xspace}
\providecommand {\calO}      {{\mathcal O}\xspace}
\providecommand {\calP}      {{\mathcal P}\xspace}
\providecommand {\calQ}      {{\mathcal Q}\xspace}
\providecommand {\calR}      {{\mathcal R}\xspace}
\providecommand {\calS}      {{\mathcal S}\xspace}
\providecommand {\calT}      {{\mathcal T}\xspace}
\providecommand {\calU}      {{\mathcal U}\xspace}
\providecommand {\calV}      {{\mathcal V}\xspace}
\providecommand {\calX}      {{\mathcal X}\xspace}
\providecommand {\calZ}      {{\mathcal Z}\xspace}

\newcommand{\mhat}[1]{\widehat{#1}}

\newcommand{\Ah}{\mhat{A}}
\newcommand{\Bh}{\mhat{B}}
\newcommand{\Ch}{\mhat{C}}
\newcommand{\Dh}{\mhat{D}}
\newcommand{\Eh}{\mhat{E}}
\newcommand{\Fh}{\mhat{F}}
\newcommand{\Gh}{\mhat{G}}
\newcommand{\Hh}{\mhat{H}}
\newcommand{\Ih}{\mhat{I}}
\newcommand{\Jh}{\mhat{J}}
\newcommand{\Kh}{\mhat{K}}
\newcommand{\Lh}{\mhat{L}}
\newcommand{\Mh}{\mhat{M}}
\newcommand{\Nh}{\mhat{N}}
\newcommand{\Oh}{\mhat{O}}
\newcommand{\Ph}{\mhat{P}}
\newcommand{\Qh}{\mhat{Q}}
\newcommand{\Rh}{\mhat{R}}
\newcommand{\Sh}{\mhat{S}}
\newcommand{\Th}{\mhat{T}}
\newcommand{\Uh}{\mhat{U}}
\newcommand{\Vh}{\mhat{V}}
\newcommand{\Wh}{\mhat{W}}
\newcommand{\Xh}{\mhat{X}}
\newcommand{\Yh}{\mhat{Y}}
\newcommand{\Zh}{\mhat{Z}}

\newcommand{\Psih}{\mhat{\Psi}}
\newcommand{\psih}{\mhat{\psi}}
\newcommand{\Phih}{\mhat{\Phi}}
\newcommand{\phih}{\mhat{\phi}}
\newcommand{\varphih}{\mhat{\varphi}}

\newcommand{\eqh}{\mhat{=}}

\newcommand{\N}{\ensuremath{\mathbb{N}}}

\newcommand{\Q}{\ensuremath{\mathbb{Q}}}

\newcommand{\R}{\ensuremath{\mathbb{R}}}

\newcommand{\perm}{\ensuremath{\pi}}

\newcommand{\allsubsets}[2]{[#1]^{#2}}

\newcommand{\pvec}[1]{\vec{#1}\mkern2mu\vphantom{#1}}

\newcommand{\kexp}[2]{\ensuremath{\exp^{#1}\hspace{-0.5mm}(#2)}}

\newcommand{\tower}[2]{\ensuremath{\text{tow}_{#1}\hspace{-0.5mm}(#2)}}

\newcommand{\klog}[2]{\ensuremath{\log^{(#1)}{\hspace{-0.5mm}(#2)}}}
\newcommand{\klogwb}[2]{\ensuremath{\log^{(#1)}{\hspace{-0.5mm}#2}}}

\newcommand{\disjointunion}{\uplus}

\providecommand{\power}[1]{\ensuremath{\calP(#1)}\xspace}

\newcommand{\restrict}[2]{#1\mspace{-3mu}\upharpoonright \mspace{-3mu}#2}

\newcommand{\isomorph}{\simeq}
\newcommand{\isomorphVia}[1]{\isomorph_{#1}}
\newcommand{\swap}[2]{id{[#1, #2]}}

\newcommand{\df}{\ensuremath{\mathrel{\smash{\stackrel{\scriptscriptstyle{
    \text{def}}}{=}}}} \;}

\newcommand{\refeq}[1]{\ensuremath{{\stackrel{\scriptstyle{
    \text{#1}}}{=}}}}

\newcommand{\longlongeq}{=\joinrel=\joinrel=\joinrel=}
\newcommand{\longeq}{=\joinrel=\joinrel=}
\newcommand{\reflongeq}[1]{\ensuremath{{\stackrel{\scriptstyle{
    \text{#1}}}{\longeq}}}}

\newcommand{\ramseyw}[1]{\ensuremath{R_{#1}}}

\makeatletter \newcommand{\auxramsey}[4]{
  \@ifmtarg{#1}{
    \@ifmtarg{#4}{
      \ensuremath{R(#2; #3)}
    }{
      \ensuremath{R^#4(#2; #3)}
    }
   }{
    \@ifmtarg{#4}{
      \ensuremath{R_{#1}(#2; #3)}
    }{
      \ensuremath{R^#4_{#1}(#2; #3)}
    }
  }
}

\newcommand{\ramsey}[3]{\auxramsey{#1}{#2}{#3}{}}
\newcommand{\homramsey}[2]{\auxramsey{}{#2}{#1}{\text{hom}}}
\newcommand{\mfoldramsey}[3]{\auxramsey{}{#2}{#1}{#3}}

\newcommand{\norder}{\prec}

\newcommand{\col}{col}

\newcommand{\property}{($\ast$)}

\newcommand{\subseq}{\sqsubseteq}

\newcommand{\derive}{\Rightarrow}
\newcommand{\rmapsto}{\rightarrow}

\newcommand{\lpath}[1][]{\ensuremath{\mathrel{\smash{\stackrel{\scriptscriptstyle{
    #1}}{\rightsquigarrow}}}}}

  \newtheorem{theorem}{Theorem}[section]
   \newtheorem{lemma}[theorem]{Lemma}
   \newtheorem{corollary}[theorem]{Corollary}
   \newtheorem{proposition}[theorem]{Proposition}
   \newtheorem{claim}{Claim}

   \newtheorem*{theorem*}{Theorem}

  \newtheorem*{maintheorem}{Main Theorem}

\newcommand{\reptheoremtitlefont}[1]{\textbf{#1}}
\newcommand{\reptheoremcontentfont}{\itshape}

\makeatletter
\newcommand{\theoremcont}[3]{
   \def\Type{#1}
   \def\Number{#2}
   \def\Label{#3}
  \@ifmtarg{#3}{
     \reptheoremtitlefont{\Type\ \Number.} \reptheoremcontentfont
   }{
    \reptheoremtitlefont{\Type\ \Number}\ \reptheoremcontentfont(\Label).
  }
}

\newenvironment{replemma}[2]{\vspace{2mm}\par\theoremcont{Lemma}{#1}{#2}}{\vspace{0mm}\par}
\newenvironment{reptheorem}[2]{\vspace{2mm}\par\theoremcont{Theorem}{#1}{#2}}{\vspace{2mm} \par }
\newenvironment{repcorollary}[2]{\vspace{2mm}\par\theoremcont{Corollary}{#1}{#2}}{\vspace{2mm} \par}

     \newtheorem{goal}{Goal}
    \theoremstyle{definition}
    \newtheorem{definition}{Definition}
    \newtheorem{example}{Example}
    
    \newtheorem {openquestion}{Open question}
    \newtheorem {question}{Question}
    \newtheorem {mainquestion}{Main question}

    \newenvironment{proofsketch}{\noindent\textsc{Proof sketch.}\enspace}{\qed \\}
    \newenvironment{proofidea}{\noindent\textsc{Proof idea.}\enspace}{\qed \\}

    \newenvironment{proofof}[1]{\vspace{2mm}\noindent\emph{Proof (of #1).}\enspace}{\qed\vspace{2mm}}
   \newenvironment{proofsketchof}[1]{\vspace{2mm}\noindent\emph{Proof sketch (of #1).}\enspace}{\qed\vspace{2mm}}

\newcommand{\eval}[3]{#1(#2/#3)}

\newcommand{\assignment}{\theta}

\newcommand{\arity}{\ensuremath{\text{Ar}}}
\newcommand{\arityFun}{\ensuremath{Ar_{\text{fun}}}}

\newcommand{\schema}{\tau}
\newcommand{\schemah}{\hat{\schema}}
\newcommand{\relSchema}{\schema_{\text{rel}}}
\newcommand{\relSchemah}{\schemah_{\text{rel}}}
\newcommand{\conSchema}{\schema_{\text{const}}}
\newcommand{\conSchemah}{\schemah_{\text{const}}}
\newcommand{\funSchema}{\schema_{\text{fun}}}
\newcommand{\funSchemah}{\schemah_{\text{fun}}}
\newcommand{\Terms}[2]{\textsc{Terms}^{#2}_{#1}} 

\newcommand{\struc}{\calS}
\newcommand{\struca}{\struc}
\newcommand{\strucb}{\calT}

\newcommand{\unaryTypes}[1]{\mathcal{UN}_{#1}}
\newcommand{\binaryTypes}[1]{\mathcal{BIN}_{#1}}
\newcommand{\naryTypes}[2]{\mathfrak{T}_{#1,#2}}

\newcommand{\nb}[3]{\calN_{#2}^{#3}(#1)}
\newcommand{\nbv}[3]{\vec \calN_{#2}^{#3}(#1)}

\newcommand{\mthen}{\rightarrow}
\newcommand{\mand}{\wedge}
\newcommand{\mor}{\vee}
\newcommand{\munion}{\cup}
\newcommand{\mintersect}{\cap}
\newcommand{\mdisjunion}{\biguplus}
\newcommand{\sem}[2]{\ensuremath{\llbracket #1\rrbracket_{#2}}} 
\newcommand{\arb}{\ensuremath{\star}}\newcommand{\generic}{\textsc{generic}}
\newcommand{\quant}{\mathbb{Q}}
\newcommand{\cquant}{\overline{\mathbb{Q}}}

\newcommand{\nd}{d}
\newcommand{\formulas}{\calC}
\newcommand{\symneg}[1]{\widehat{#1}}

\newcommand{\type}[2]{\ensuremath{\langle #1, #2 \rangle}}
\newcommand{\stype}[3]{\ensuremath{\langle #1, #2 \rangle_{#3}}}

\newcommand{\behaveEqual}[1]{\approx_{#1}}

\newcommand{\types}[2]{types_{#1}(#2)}
\newcommand{\numTypes}[2]{|\types{{#1}}{#2}|}

\newcommand{\eqtype}{\epsilon}

\newcommand{\db}{\calD}
\newcommand{\inp}{\calI}
\newcommand{\aux}{\calA}
\newcommand{\builtin}{\calB}
\newcommand{\domain}{D}

\newcommand{\query}{\calQ}
\newcommand{\cq}{\calC}

\newcommand{\querys}{Q}

\newcommand{\ans}[2]{\mtext{ans}(#1, #2)}

\newcommand{\updates}{\ensuremath{\Delta}}
\newcommand{\abstrDel}{\ensuremath{\updates_{Del}}}
\newcommand{\abstrIns}{\ensuremath{\updates_{Ins}}}
\newcommand{\abstrUpd}{\ensuremath{\updates}}

\newcommand{\init}{\mtext{Init}\xspace}

\newcommand{\ins}{\mtext{ins}}
\newcommand{\del}{\mtext{del}}

\newcommand{\insertdescr}[2]{\textbf{Insertion of \ensuremath{#2} into \ensuremath{#1}.}}
\newcommand{\deletedescr}[2]{\textbf{Deletion of \ensuremath{#2} from \ensuremath{#1}.}}

\newcommand{\state}{\struc}

\newcommand{\inpSchema}{\schema_{\text{in}}}
\newcommand{\auxSchema}{\schema_{\text{aux}}}
\newcommand{\eqSchema}{\schema_{=}}
\newcommand{\builtinSchema}{\schema_{\text{bi}}}

\newcommand{\auxInit}{\init_{\text{aux}}}
\newcommand{\builtinInit}{\init_{\text{bi}}}

\providecommand{\prog}{\ensuremath{\calP}\xspace}
\newcommand{\progb}{\ensuremath{Q}\xspace}

\newcommand{\updateDB}[2]{\ensuremath{#1(#2)}}
\newcommand{\updateState}[3]{\ensuremath{#1_{#2}(#3)}}
\newcommand{\updateRelation}[4]{\restrict{\ensuremath{{#1}_{#2}(#3)}}{#4}}

\newcommand{\transition}[3]{\ensuremath{{#1} \xrightarrow{#2}{#3}}}

\makeatletter \newcommand{\uf}[4]{
  \@ifmtarg{#4}{
    \ensuremath{\phi^{#1}_{#2}(#3)}
   }{
    \ensuremath{\phi^{#1}_{#2}(#3; #4)}
  }
}
\newcommand{\huf}[4]{
  \@ifmtarg{#4}{
    \ensuremath{\widehat{\phi}^{#1}_{#2}(#3)}
   }{
    \ensuremath{\widehat{\phi}^{#1}_{#2}(#3; #4)}
  }
}

\newcommand{\ufb}[4]{
  \@ifmtarg{#4}{
    \ensuremath{\psi^{#1}_{#2}(#3)}
   }{
    \ensuremath{\psi^{#1}_{#2}(#3; #4)}
  }
}

\newcommand{\ufbwa}[2]{
  \ensuremath{\psi^{#1}_{#2}}
}

\newcommand{\ufwa}[2]{
  \ensuremath{\phi^{#1}_{#2}}
}

    \makeatletter   \newcommand{\ufsubstitute}[5]{
    \@ifmtarg{#5}{
      \ensuremath{\phi^{#2}_{#3}[#1](#4)}
    }{
      \ensuremath{\phi^{#2}_{#3}[#1](#4; #5)}
    }
  }

    \makeatletter   \newcommand{\ufsubstitutewa}[3]{
      \ensuremath{\phi^{#2}_{#3}[#1]}
  }
  \makeatletter   \newcommand{\substitutewa}[2]{
      \ensuremath{#1[#2]}
  }

\newcommand{\ut}[4]{
  \@ifmtarg{#4}{
    \ensuremath{t^{#1}_{#2}(#3)}
   }{
    \ensuremath{t^{#1}_{#2}(#3; #4)}
  }
}

\newcommand{\utw}[3]{
  \ensuremath{t^{#1}_{#2}(#3)}
}

\newcommand{\utwa}[2]{\ensuremath{t^{#1}_{#2}}}
\newcommand{\ite}[3]{
  \@ifmtarg{#1}{
    \ensuremath{\mtext{ITE}}
   }{
    \mtext{ITE}(#1,#2,#3)  
  }
}

\newcommand{\itewa}{
    \ensuremath{\mtext{ite}}
}

\providecommand{\nc}{\newcommand}
\providecommand{\rnc}{\renewcommand}
\providecommand{\pc}{\providecommand}

\renewcommand{\labelenumi}{(\alph{enumi})}

\newcommand{\Erdos}{Erd\H{o}s}

\newcommand{\quotes}[1]{``#1''}

\newcommand{\mbold}[1]{\textbf{\hyperpage{#1}}}
\newcommand{\mitalic}[1]{{\it #1}}

\newcommand{\bfindex}[1]{\index{#1|mbold}}
\newcommand{\itindex}[1]{\index{#1|it}}

\newcommand{\defindex}[1]{\bfindex{#1}}

\ifcomments
\nc{\commentbox}[1]{\noindent\framebox{\parbox{\linewidth}{#1}}}
\nc{\todo}[1]{\ \\ {\color{red} \fbox{\parbox{\linewidth}{{\sc
          ToDo}:\\  #1}}}}

\setlength{\marginparwidth}{2.5cm}
\setlength{\marginparsep}{3pt}

\newcounter{CommentCounter}
\newcommand{\acomment}[2]{\ \\ \fbox{\parbox{\linewidth}{{\sc #1}: #2}}}
\newcommand{\mcomment}[2]{{\color{blue}(#1)}\footnote{#1: #2}}                                 \else
\nc{\commentbox}[1]{}
\newcommand{\mcomment}[2]{}
\newcommand{\acomment}[2]{}
\fi

\ifchanges

\newcommand{\loldnew}[3]{\commentbox{{\textcolor{blue}{\setlength{\fboxsep}{1pt}\fbox{\small
          #1}}} \textcolor{red}{\footnotesize #2}}
  \textcolor{blue}{#3}}
\setul{}{0.2mm}
\setstcolor{red}
\newcommand{\oldnew}[3]{{\textcolor{blue}{\setlength{\fboxsep}{1pt}\fbox{\small
        #1}}} \st{\footnotesize #2} \textcolor{blue}{#3}}

\else
\newcommand{\loldnew}[3]{#3}
\newcommand{\oldnew}[3]{#3}
\fi

\nc{\tzm}[1]{\mcomment{TZ}{#1}}
\nc{\tsm}[1]{\mcomment{TS}{#1}}
\nc{\tz}[1]{\acomment{TZ}{#1}}
\nc{\thz}[1]{\acomment{TZ}{#1}}
\nc{\ts}[1]{\acomment{TS}{#1}}

\nc{\tzon}[2][]{\oldnew{TZ}{#1}{#2}} 
\nc{\tson}[2][]{\oldnew{TS}{#1}{#2}}

\nc{\tzlon}[2][]{\loldnew{TZ}{#1}{#2}} 
\nc{\tslon}[2][]{\loldnew{TS}{#1}{#2}}

\newcommand{\columnWidth}{11cm}

\newcommand{\substruclemma}{Substructure Lemma\xspace}
\newcommand{\First}{\mtext{First}}
\newcommand{\List}{\mtext{List}}
\newcommand{\Last}{\mtext{Last}}
\newcommand{\In}{\mtext{In}}
\newcommand{\Out}{\mtext{Out}}
\newcommand{\Empty}{\mtext{Empty}}

\newcommand{\Odd}{\mtext{Odd}}
\newcommand{\odd}{\text{odd}}
\newcommand{\even}{\text{even}}

\newcommand{\Counter}{\mtext{Cnt}}
\newcommand{\isEmpty}{\mtext{Empty}}
\newcommand{\Zero}{\mtext{Zero}}

\newcommand{\congruent}[2]{\sim_{#1, #2}}

\newcommand{\Succ}{\mtext{Succ}}
\newcommand{\Pred}{\mtext{Pred}}
\newcommand{\Max}{\mtext{Max}}
\newcommand{\numEdges}{\#\mtext{edges}}
\newcommand{\numNodes}{\#\mtext{nodes}}

\newcommand{\shortVersion}[1]{#1}
\newcommand{\longVersion}[2]{}

\newcommand{\apptheoremtitlefont}[1]{\textbf{#1}}
\newcommand{\apptheoremcontentfont}{\itshape}

\newcommand{\apponlystartmarker}{ $\blacktriangleright\blacktriangleright\blacktriangleright$ }
\newcommand{\apponlyendmarker}{ $\blacktriangleleft\blacktriangleleft\blacktriangleleft$ }
\newcommand{\apprepetitionstartmarker}{ $\blacktriangleright\blacktriangleright\blacktriangleright$ }
\newcommand{\apprepetitionendmarker}{ $\blacktriangleleft\blacktriangleleft\blacktriangleleft$ }

\newcommand{\initialAppendix}{
  \section*{Appendix}

  \setcounter{section}{1}
   \renewcommand{\thesection}{\Alph{section}}
   \counterwithin{theorem}{section}
   \counterwithin{lemma}{section}
   \counterwithin{corollary}{section}
   \counterwithin{definition}{section}
   \counterwithin{example}{section}

  In the appendix we give the proofs that have been omitted in the main text. For proofs that are partially present in the main article, we repeat the full proof and its context. Parts that are only repeated are marked by \apprepetitionstartmarker and \apprepetitionendmarker.  
  
}
  
\newcommand{\writeAppendix}{\initialAppendix}
\newcommand{\toAppendix}[1]{
  \makeatletter
   \g@addto@macro\writeAppendix{#1}
  \makeatother
}

\newcommand{\toMainAndAppendix}[1]{
      #1  \toAppendix{      \apprepetition{#1} \par
  }
}

\newcommand{\atheorem}[2]{
  \begin{theorem}\label{#1}    #2  \end{theorem}  \toAppendix{    \begin{apptheorem}{\ref{#1}}{}
      #2    \end{apptheorem}  }
}

\newcommand{\alemma}[2]{
  \begin{lemma}\label{#1}    #2  \end{lemma}  \toAppendix{    \begin{applemma}{\ref{#1}}{}      #2    \end{applemma}  }
}

\newcommand{\aproof}[3]{  \@ifmtarg{#2}{}{    \begin{proof}      #1      #2    \end{proof}  }
  \toAppendix{    \begin{proof}      \@ifmtarg{#1}{}{\apprepetition{#1}} \par
      #3
    \end{proof}  }
}

\newcommand{\aproofsketch}[3]{  \@ifmtarg{#2}{}{    \begin{proofsketch}      #1      #2    \end{proofsketch}  }
  \toAppendix{    \begin{proof}      \@ifmtarg{#1}{}{\apprepetition{#1}} \par
      #3
    \end{proof}  }
}

\newcommand{\shortOrLong}[2]{
  \shortVersion{#1}
  \longVersion{#2}
}

\newcommand{\apponlystart}{
    \apponlystartmarker
    }
\newcommand{\apponlyend}{
    \apponlyendmarker
    }

\newcommand{\apprepetition}[1]{
  \apprepetitionstartmarker #1 \apprepetitionendmarker
}

\tikzstyle{mnode}=[
  circle,
  fill=\mnodefillcolor, 
  draw=\mnodedrawcolor,
  minimum size=6pt, 
  inner sep=0pt
]

\tikzstyle{mnodeinvisible}=[
  minimum size=6pt, 
  inner sep=0pt
]

\tikzstyle{invisible}=[
  minimum size=0pt, 
  inner sep=0pt
]

\tikzstyle{invisiblel}=[
  minimum size=10pt, 
  inner sep=0pt
]

\tikzstyle{invisibleEdge}=[
  transparent
]

\tikzstyle{nameNode}=[
  font=\scriptsize
]

\tikzstyle{namingNode}=[
  font=\normalsize
]

\tikzstyle{mEdge}=[
  -latex',   thick, 
  shorten >=3pt, 
  shorten <=3pt,
  draw=black!80,
]

\tikzstyle{dDashedEdge}=[
  -latex',   thick, 
  shorten >=3pt, 
  shorten <=3pt,
  draw=black!80,
  dashed
]

\tikzstyle{dEdge}=[
  -latex',   thick, 
  shorten >=3pt, 
  shorten <=3pt,
  draw=black!80,
]

\tikzstyle{dhEdge}=[
  -latex',   thick, 
  shorten >=3pt, 
  shorten <=3pt,
  draw=black!80,
]
\tikzstyle{uEdge}=[
  thick, 
  shorten >=3pt, 
  shorten <=3pt,
  draw=black!80,
]
\tikzstyle{uhEdge}=[
  thick, 
  shorten >=3pt, 
  shorten <=3pt,
  draw=black!80,
]

\tikzstyle{cEdge}=[
  ultra thick, 
  shorten >=3pt, 
  shorten <=3pt,
  draw=black!80,
]

\tikzstyle{dotsEdge}=[
  thick, 
  loosely dotted, 
  shorten >=7pt, 
  shorten <=7pt
]

\tikzstyle{snakeEdge}=[
  ->, 
  decorate, 
  decoration={snake,amplitude=.4mm,segment length=2.5mm,post length=0.5mm},
]

\tikzstyle{snakeEdgea}=[
  ->, 
  decorate, 
  decoration={snake,amplitude=.4mm,segment length=3mm,post length=0.5mm}
]

\newcommand{\redEdge}{\tikz{
  \node (tmpa) at (-0.4,0)[mnode]{};
  \node (tmpb) at (0.4,0)[mnode]{};
  \draw [cEdge, draw=red] (tmpa) to (tmpb);
}}

\newcommand{\yellowEdge}{\tikz{
  \node (tmpa) at (-0.4,0)[mnode]{};
  \node (tmpb) at (0.4,0)[mnode]{};
  \draw [cEdge, draw=yellow] (tmpa) to (tmpb);
}}

\newcommand{\blueEdge}{\tikz{
  \node (tmpa) at (-0.4,0)[mnode]{};
  \node (tmpb) at (0.4,0)[mnode]{};
  \draw [cEdge, draw=blue] (tmpa) to (tmpb);
}}

\newcommand{\violetEdge}{\tikz{
  \node (tmpa) at (-0.4,0)[mnode]{};
  \node (tmpb) at (0.4,0)[mnode]{};
  \draw [cEdge, draw=violet] (tmpa) to (tmpb);
}}

\newcommand{\textnode}[1]{      \begin{minipage}{100pt}    #1
      \end{minipage}}

\newcommand{\classNode}[2]{
    \begin{scope}[shift={#1}]
      \node[text depth=0.7cm, anchor=north] (tmp) at (0,0) {
  \begin{minipage}{170pt}    #2
  \end{minipage}      };
    \end{scope}
}

\tikzstyle{class rectangle}=[
  draw=black,
  inner sep=0.2cm,
  rounded corners=5pt,
  thick
]

\tikzstyle{mline}=[
  draw=black,
  inner sep=0.2cm,
  rounded corners=5pt,
  thick
]

\tikzstyle{mainclass rectangle}=[
  draw=blue,
  inner sep=0.2cm,
  rounded corners=5pt,
  very thick
]

\newcommand{\redpoint}{\tikz{\node (tmp) at (-1.5,1.6)[mnode, fill=red]{};}}
\newcommand{\yellowpoint}{\tikz{\node (tmp) at (-1.5,1.6)[mnode, fill=yellow]{};}}
\newcommand{\bluepoint}{\tikz{\node (tmp) at (-1.5,1.6)[mnode, fill=blue]{};}}
\newcommand{\violetpoint}{\tikz{\node (tmp) at (-1.5,1.6)[mnode, fill=violet]{};}}

\newcommand{\bgcolor}{black!5}
\newcommand{\substructurefillcolor}{blue!40}
\newcommand{\substructureufillcolor}{blue!60!red!40}
\newcommand{\substructuredrawcolor}{blue!80}
\newcommand{\substructureudrawcolor}{blue!60!red!80}

\newcommand{\structurefillcolor}{blue!20}
\newcommand{\structureufillcolor}{blue!60!red!20}
\newcommand{\structuredrawcolor}{blue!40}
\newcommand{\structureudrawcolor}{blue!60!red!40}

\newcommand{\mnodedrawcolor}{black!80}
\newcommand{\mnodefillcolor}{black!40}

\newcommand{\mnc}{\mnodefillcolor}
\newcommand{\mgrey}{\mnodefillcolor}
\newcommand{\mblue}{blue}
\newcommand{\mred}{red}
\newcommand{\myellow}{yellow}
\newcommand{\mviolet}{violet}

\newcommand{\mopacity}{1.0}
\newcommand{\opac}{0.2}

\pgfdeclarelayer{background}
\pgfdeclarelayer{substructure}
\pgfdeclarelayer{edges}
\pgfdeclarelayer{foreground}
\pgfsetlayers{background,substructure,edges,main,foreground}

\tikzstyle{background rectangle}=[
  fill=black!5,
  draw=black!20,
  inner sep=0.2cm,
  rounded corners=5pt
]

\tikzstyle{delEdge}=[
  -latex', 
  semithick, 
  shorten >=3pt, 
  shorten <=3pt,
  dotted,
  draw=black!80,
]

\newcommand{\deleteuedge}[2]
{    \draw [deluEdge] #1 to #2;
    \smallcross{($#1!0.5!#2$)}
}
\newcommand{\deleteedge}[2]
{    \draw [delEdge] #1 to #2;
    \smallcross{($#1!0.5!#2$)}
}

\tikzstyle{insEdge}=[
  -latex', 
  semithick, 
  shorten >=3pt, 
  shorten <=3pt,
  dotted,
  draw=black!80,
]

\tikzstyle{dbbackground rectangle}=[
  fill=blue!5,
  draw=blue!20,
  inner sep=0.2cm,
  rounded corners=5pt
]
\tikzstyle{auxbackground rectangle}=[
  fill=green!20,
  draw=green!60,
  inner sep=0.2cm,
  rounded corners=5pt
]

\newcommand{\drawyPoints}[1]{
   \foreach \n/\t/\p/\l in #1 {
      \node (\n) at \p [mnode, label={[labelbg]\l:\t}] {};
    }
}

\newcommand{\cross}[1]{
  \begin{scope}[shift={#1}]
    \draw [very thick, draw=red](0.5,0.5) -- (-0.5,-0.5);
    \draw [very thick, draw=red](-0.5,0.5) -- (0.5,-0.5);
  \end{scope}
}
\newcommand{\smallcross}[1]{
  \begin{scope}[shift={#1}]
    \draw [ultra thick, draw=red](0.25,0.25) -- (-0.25,-0.25);
    \draw [ultra thick, draw=red](-0.25,0.25) -- (0.25,-0.25);
  \end{scope}
}

\newcommand{\mytable}[1]{
  \begin{scope}[shift={#1}]
    \draw [fill=white](-1.8,-2) rectangle (1.8,2);
    \draw(-1.8,1.5)--(1.8,1.5);
    \draw(-0.9,-2)--(-0.9,2);
    \draw(0,-2)--(0,2);
    \draw(0.9,-2)--(0.9,2);
  \end{scope}
}

\newcommand{\database}[1]{
  \begin{scope}[shift={#1}]
    \mytable{(1,0.7)};
    \mytable{(-1,0)};
    \mytable{(0.2,-0.8)};
  \end{scope}
}

\newcommand{\databasewithbg}[1]{
  \database{#1}
  \begin{pgfonlayer}{substructure}
    \begin{scope}[shift={#1}]
  \draw [auxbackground rectangle](-4.5,-4) rectangle (4.5,4);
    \end{scope}
  \end{pgfonlayer}
}

\newcommand{\multidedge}[4]{
    \node (1) at #1[invisible] {};
    \node (2) at #2[invisible] {};
    \node (3) at #3[invisible] {};
    \node (4) at #4[invisible] {};

     \draw[thick] (1) to [bend left] (3);
     \draw[thick] (2) to [bend right] (3);
     \draw[->, thick] (3) to [thick] (4);
}

\newcommand{\auxUpdateEdge}[1]{
    \begin{scope}[shift={#1}]
      \multidedge{(-4,-5)}{(1,-2)}{(-1,-4)}{(0,-5.5)}
    \end{scope}
}

\tikzstyle{substructure}=[
  fill=\substructurefillcolor,
  draw=\substructuredrawcolor,
  inner sep=0.2cm,
  rounded corners=5pt
]
\tikzstyle{substructureu}=[
  fill=\substructureufillcolor,
  draw=\substructureudrawcolor,
  inner sep=0.2cm,
  rounded corners=5pt
]
\tikzstyle{substructurewoborder}=[
  fill=\substructurefillcolor,
  draw=\substructurefillcolor,
  inner sep=0.2cm,
  rounded corners=5pt
]

\tikzstyle{structure}=[
  fill=\structurefillcolor,
  draw=\structuredrawcolor,
  inner sep=0.2cm,
  rounded corners=5pt
]

\tikzstyle{structureu}=[
  fill=\structureufillcolor,
  draw=\structureudrawcolor,
  inner sep=0.2cm,
  rounded corners=5pt
]

\tikzstyle{labelsubstruc}=[
  inner sep=0pt,
  outer sep=1.5pt,
]

\tikzstyle{labelbg}=[
  inner sep=1.5pt,
]

\newcommand{\picDynPropLowerBoundClique}{
  \begin{tikzpicture}[
      xscale=1,
      yscale=1,
      font=\large,
      show background rectangle,
    ]

        \draw [mline] (-1.5, 8) rectangle (10.5, 5);
        \draw [mline] (-1, 7.5) rectangle (8.5, 5.5);
          \draw [mline] (-1.5, 2.5) rectangle (10.5, 4.5);

     \draw [mline] (4,2.5) -- (5,4.5);
    
    \node (A) at (11, 6.5) [nameNode,font=\Large] {$A$};
    \node (C) at (12, 3.5) [nameNode,font=\Large] {$C \df \allsubsets{A}{k+1}$};
     \node (Ap) at (8, 7) [nameNode,font=\Large] {$A'$};
     \node (B) at (4,4) [nameNode,font=\Large] {$B$};
     \node (Bp) at (5,3) [nameNode,font=\Large] {$B'$};    

    \node (b1) at (0, 7)[mnode, label={[labelbg]left:$b_1$}] {};
    \node (b2) at (1, 7) [mnode, label={[labelbg]left:$b_2$}] {};
    \node (bk) at (2.5, 7) [mnode, label={[labelbg]right:$b_{k+1}$}] {};
    \node (bp1) at (4,6) [mnode , label={[labelbg]left:$b'_1$}] {};
    \node (bp2) at (5,6) [mnode, label={[labelbg]left:$b'_2$}] {};
    \node (bpk) at (6.5,6) [mnode, label={[labelbg]right:$b'_{k+1}$}] {};
    \node (b12k) at (-0.25,3) [mnode, label={[labelbg]right:$b = \{b_1, b_2, \ldots, b_{k+1}\}$}] {};
    \node (bp12k) at (5.75,4) [mnode, label={[labelbg]right:$b' = \{b'_1, b'_2, \ldots, b'_{k+1}\}$}] {};

     \draw [dotsEdge, ultra thick] (b2) to (bk);
     \draw [dotsEdge, ultra thick] (bp2) to (bpk);

    \draw [mEdge] (b12k) to (b1);
    \draw [mEdge] (b12k) to (b2);
    \draw [mEdge] (b12k) to (bk);
  \end{tikzpicture}
}

\newcommand{\picDynPropLowerBoundEAFO}{
  \begin{tikzpicture}[
      xscale=1,
      yscale=1,
      font=\large,
      show background rectangle,
    ]

        \draw [mline] (-1.5, 8) rectangle (10.5, 5);
        \draw [mline] (-1, 7.5) rectangle (8.5, 5.5);
          \draw [mline] (-1.5, 2.5) rectangle (10.5, 4.5);

     \draw [mline] (4,2.5) -- (5,4.5);
    
    \node (A) at (11, 6.5) [nameNode,font=\Large] {$A$};
    \node (C) at (12, 3.5) [nameNode,font=\Large] {$C \df \allsubsets{A}{k+1}$};
     \node (Ap) at (8, 7) [nameNode,font=\Large] {$A'$};
     \node (B) at (4,4) [nameNode,font=\Large] {$B$};
     \node (Bp) at (5,3) [nameNode,font=\Large] {$B'$};    

    \node (b1) at (0, 7)[mnode, label={[labelbg]left:$b_1$}] {};
    \node (b2) at (1, 7) [mnode, label={[labelbg]left:$b_2$}] {};
    \node (bk) at (2.5, 7) [mnode, label={[labelbg]right:$b_{k+1}$}] {};
    \node (bp1) at (4,6) [mnode , label={[labelbg]left:$b'_1$}] {};
    \node (bp2) at (5,6) [mnode, label={[labelbg]left:$b'_2$}] {};
    \node (bpk) at (6.5,6) [mnode, label={[labelbg]right:$b'_{k+1}$}] {};
    \node (b12k) at (-0.25,3) [mnode, label={[labelbg]right:$b = \{b_1, b_2, \ldots, b_{k+1}\}$}] {};
    \node (bp12k) at (5.75,4) [mnode, label={[labelbg]right:$b' = \{b'_1, b'_2, \ldots, b'_{k+1}\}$}] {};

    \node (s) at (4, 1.5)[mnode, label={[labelbg]below:$s$}] {};
    \node (t) at (4, 9) [mnode, label={[labelbg]above:$t$}] {};

     \draw [dotsEdge, ultra thick] (b2) to (bk);
     \draw [dotsEdge, ultra thick] (bp2) to (bpk);

    \draw [mEdge] (b12k) to (b1);
    \draw [mEdge] (b12k) to (b2);
    \draw [mEdge] (b12k) to (bk);
    \draw [mEdge,  shorten >=10pt] (s) to (b12k);
    \draw [mEdge, dashed] (b1) to (t);
    \draw [mEdge, dashed] (b2) to (t);
    \draw [mEdge, dashed] (bk) to (t);
    
  \end{tikzpicture}
}

\newcommand{\picDynPropLowerBoundConstructiona}{
  \begin{tikzpicture}[
      xscale=0.8,
      yscale=0.6,
      font=\large,
      show background rectangle,
    ]

        \draw [mline] (-1, -0.5) rectangle (5, 7);
        \draw [mline] (0, 1) rectangle (4.5, 6.5);
        \draw [mline] (6.5, -0.5) rectangle (12.5, 7);

    \draw [mline] (6.5,2) -- (12.5,3);
    
    \node (A) at (1.5,-1.2) [nameNode,font=\Large] {$A$};
    \node (C) at (9.5,-1.2) [nameNode,font=\Large] {$C \df \allsubsets{A}{k+1}$};
    \node (Ap) at (0.5, 6.0) [nameNode,font=\Large] {$A'$};
    \node (B) at (7.5,2.8) [nameNode,font=\Large] {$B$};
    \node (Bp) at (11.5,2.2) [nameNode,font=\Large] {$B'$};    

    \node (b1) at (3.5, 6)[mnode, label={[labelbg]left:$b_1$}] {};
    \node (b2) at (3.5,5) [mnode, label={[labelbg]left:$b_2$}] {};
    \node (bk) at (3.5,3.0) [mnode, label={[labelbg]left:$b_{k+1}$}] {};
    \node (bp1) at (1.5,4.5) [mnode , label={[labelbg]left:$b'_1$}] {};
    \node (bp2) at (1.5,3.5) [mnode, label={[labelbg]left:$b'_2$}] {};
    \node (bpk) at (1.5,1.5) [mnode, label={[labelbg]left:$b'_{k+1}$}] {};
    \node (b12k) at (7.5,5) [mnode, label={[labelbg]right:$b = \{b_1, b_2, \ldots, b_{k+1}\}$}] {};
    \node (bp12k) at (7.5,0.5) [mnode, label={[labelbg]right:$b' = \{b'_1, b'_2, \ldots, b'_{k+1}\}$}] {};

     \draw [dotsEdge] (b2) to (bk);
     \draw [dotsEdge] (bp2) to (bpk);

    \draw [mEdge] (b12k) to (b1);
    \draw [mEdge] (b12k) to (b2);
    \draw [mEdge] (b12k) to (bk);
  \end{tikzpicture}
}

\newcommand{\picDynPropUpperBoundConstructionA}{
  \begin{tikzpicture}[
      xscale=0.8,
      yscale=0.8,
      font=\large,
      show background rectangle,
    ]

     \node (G) at (-0.5,3) [nameNode,font=\large] {$G$:};

    \node (a1) at (1, -0.5)[mnode, label={[labelbg]below:$a_1$}] {};
    \node (a2) at (2,1) [mnode, label={[labelbg]right:$a_2$}] {};
    \node (a3) at (0,1) [mnode, label={[labelbg]left:$a_{3}$}] {};
    \node (a4) at (1,2.5) [mnode , label={[labelbg]above:$a_4$}] {};

    \draw [dDashedEdge] (a3) to (a1);
    \draw [dEdge] (a3) to (a2);
    \draw [dEdge] (a2) to (a4);
  \end{tikzpicture}
}
\newcommand{\picDynPropUpperBoundConstructionB}{
  \begin{tikzpicture}[
      xscale=0.8,
      yscale=0.8,
      font=\large,
      show background rectangle,
    ]

     \node (G) at (-0.5,3) [nameNode,font=\large] {$H$:};

    \node (a1) at (1, -0.5)[mnode, label={[labelbg]below:$x_1$}] {};
    \node (a2) at (2,1) [mnode, label={[labelbg]right:$x_2$}] {};
    \node (a3) at (0,1) [mnode, label={[labelbg]left:$x_{3}$}] {};
    \node (a4) at (1,2.5) [mnode , label={[labelbg]above:$x_4$}] {};

    \draw [dEdge] (a3) to (a1);
    \draw [dEdge] (a3) to (a2);
    \draw [dEdge] (a2) to (a4);
    \draw [dEdge] (a1) to (a2);
  \end{tikzpicture}
}

  \newcommand{\substructpic}[2]{
      \draw[substructure] (0,0) ellipse (0.8 and 0.8);
      \node (tmp) at (0,0.3) {#1};
      {[shift={(0,-0.3)}] #2}
   }
  \newcommand{\structpica}[1]{
       \draw[structure] (-1.0, -1.0) rectangle (1.2, 1.9);
      \node (tmp) at (0,1.3) {#1};
   }
  \newcommand{\structpicb}[1]{
      \node[structure, isosceles triangle, shape border rotate=90, minimum height=2.8cm, minimum width=3.2cm, anchor=lower side, isosceles triangle stretches
] at (0,-1.0) {};
      \node (tmp) at (0,1.3) {#1};
   }

  \newcommand{\substructpicu}[2]{
      \draw[substructureu] (0,0) ellipse (0.8 and 0.8);
      \node (tmp) at (0,0.3) {#1};
      {[shift={(0,-0.3)}] #2}
   }

  \newcommand{\structpicua}[1]{
      \draw[structureu] (-1.0, -1.0) rectangle (1.2, 1.9);
      ] at (0,-1.0) {};
      \node (tmp) at (0,1.3) {#1};
   }
  \newcommand{\structpicub}[1]{
     \node[structureu, isosceles triangle, shape border rotate=90, minimum height=2.8cm, minimum width=3.2cm, anchor=lower side, isosceles triangle stretches
] at (0,-1.0) {};
      \node (tmp) at (0,1.3) {#1};
   }

\newcommand{\picsubstructure}{
    \begin{tikzpicture}[
       xscale=1.0,
       yscale=1.0,
    ]
    \node[invisible] (tmp) at (-2,0) {};

   \begin{scope}[shift={(0,4.5)}]
      \structpica{$S$}
      \node (tmp) at (-1.5,1.6){$\calS$};
      \substructpic{$A$}{\node (lo) at (0.2,-0.2)[mnode,label={left:$\vec a$}] {};}
   \end{scope}
   \begin{scope}[shift={(5,4.5)}] 
      \structpicb{$T$}
      \node (tmp) at (+1.2,1.6){$\calT$};
      \substructpic{$B$}{\node (ro) at (-0.4,-0.2)[mnode,label={right:$\pi(\vec a)$}] {};}
    \end{scope}
    \draw [dEdge, very thick] (lo) to node[above]{$\cong$}node[below]{$\pi$}(ro);

      \draw [dEdge](0, 3.5) -- node[right]{$\alpha = \delta(\vec a)$}(0, 2);
      \draw [dEdge](5, 3.5) -- node[right]{$\beta = \delta(\pi(\vec a))$}(5, 2); 

  \begin{scope}[shift={(0,0)}] 
    \structpicua{$S$}
    \node (tmp) at (-1.7,1.6){$\updateState{\prog}{\alpha}{\calS}$};
    \substructpicu{$A$}{\node (lu) at (0.2,0)[invisible]{};}
  \end{scope}

  \begin{scope}[shift={(5,0)}] 
    \structpicub{$T$}
    \node (tmp) at (+1.2,1.6){$\updateState{\prog}{\beta}{\calT}$};
    \substructpicu{$B$}{\node (ru) at (-0.3,0)[invisible]{};}
  \end{scope}
    
      \draw [dEdge, very thick] (lu) to node[above]{$\cong$}node[below]{$\pi$}(ru);

  \end{tikzpicture}
}    
     \title{The Dynamic Descriptive Complexity of $k$-Clique\footnote{An extended abstract of this work appeared in the proceedings of the conference Mathematical Foundations of Computer Science 2014 (MFCS 2014)\cite{Zeume14}. The author acknowledges the financial support by the German DFG under grant SCHW 678/6-1.}}

    \author[ ]{Thomas Zeume} 
    \affil[ ]{TU Dortmund University}
  \affil[ ]{\textit {thomas.zeume@tu-dortmund.de}}

    \maketitle

  \begin{abstract}
    In this work the dynamic descriptive complexity of the $k$-clique query is studied. It is shown that when edges may only be inserted then $k$-clique can be maintained by a quantifier-free update program of arity $k-1$, but it cannot be maintained by a quantifier-free update program of arity $k-2$ (even in the presence of unary auxiliary functions). This establishes an arity hierarchy for graph queries for quantifier-free update programs under insertions. The proof of the lower bound uses upper and lower bounds for Ramsey numbers.        
  \end{abstract}

  \sloppy

  \section{Introduction}\label{section:introduction}
    The $k$-clique query --- does a given graph contain a $k$-clique? --- can be expressed by an existential first-order formula with $k$ quantifiers. In this work we study the descriptive complexity of the $k$-clique query in a setting where edges may be inserted dynamically into a graph. In particular we are interested in lower bounds for the resources necessary to express this query dynamically.

The dynamic descriptive complexity framework (short: dynamic complexity), independently introduced by Dong, Su and Topor \cite{DongT92, DongS93} and Patnaik and Immerman \cite{PatnaikI94}, models the setting of dynamically changing graphs. For a graph subject to changes, auxiliary relations are maintained with the intention to help answering a query $\query$. When an insertion (or, in the general setting, a deletion) of an edge occurs, every auxiliary relation is updated through a first-order query that can refer to both the graph itself and the auxiliary relations. The query $\query$ is maintained by such a program, if one designated auxiliary relation always stores the current query result. The class of all queries maintainable by first-order update programs is called $\DynFO$\footnote{In this work we stick to the specific framework introduced by Patnaik and Immerman. The main result also holds in the framework of Dong, Su and Topor, even though both frameworks differ in details.}. 

Since $k$-clique can be expressed in existential first-order logic, it can be trivially maintained by a first-order update program. Therefore for characterizing the precise dynamic complexity of this query we need to look at fragments of $\DynFO$. It turns out that $k$-clique can still be maintained under insertions when the update formulas are not allowed to use quantifiers at all and auxiliary relations may only have restricted arity. We obtain the following characterization.

\begin{itemize}
 \item [\textbf{Main result:}] When only edge insertions are allowed then $k$-clique ($k \geq 3$) can be maintained by a quantifier-free update program of arity $k-1$, but it cannot be maintained by a quantifier-free update program of arity $k-2$.
\end{itemize}

Actually we prove that every property expressible by a positive existential first-order formula with $k$ quantifiers and, possibly, negated equality atoms can be maintained by a $(k-1)$-ary quantifier-free program under insertions.

In order to understand why the lower bound contained in the above result is interesting, we shortly discuss the status quo of lower bound methods for the dynamic complexity framework. Up to now very few lower bounds are known; all of them for fragments of $\DynFO$ obtained by either bounding the arity of the auxiliary relations or by restricting the usage of quantifiers (or by restricting both). Usually those bounds have been stated only for the setting where both insertions and deletions are allowed. We emphasize that our lower bound for the insertion-only setting immediately transfers to this more general setting.  

The study of bounded arity auxiliary relations was started by Dong and Su~\cite{DongS98}. They exhibited concrete graph queries that cannot be maintained in unary $\DynFO$, and they showed that \DynFO has an arity hierarchy for general (that is non-graph) queries. Both results rely on previously obtained static lower bounds.

Hesse started the study of the quantifier-free fragment of $\DynFO$ (short: $\DynProp$) in~\cite{Hesse03}. Although this fragment appears to be  rather weak at first glance, deterministic reachability \cite{Hesse03} and regular languages \cite{GeladeMS12} can be maintained in $\DynProp$. In  \cite{GeladeMS12},  Gelade et al. also provided first lower bounds. They proved that non-regular languages as well as the alternating reachability problem cannot be maintained in this fragment. The use of very restricted graphs in the proof of the latter result implies that there is a $\QFO[\exists^*\forall^*\exists^*]$-definable query that cannot be maintained in $\DynProp$. In \cite{ZeumeS15reach} it was shown that reachability and $3$-clique cannot be maintained in the binary quantifier-free fragment of $\DynFO$.

In general, it is a difficult task to prove lower bounds in the dynamic complexity setting; even when update formulas are restricted to the quantifier-free fragment of first-order logic. We are not at the point where we can, when given a query, apply a set of tools in order to prove that the query cannot be maintained in \DynProp. Finding more queries that cannot be maintained in \DynProp seems to be a reasonable approach towards finding more generic proof methods.

The lower bound provided by the main result follows this approach and is interesting in two ways. First, it exhibits, for every $k$, a query in $\QFO[\exists^{k}]$ that cannot be maintained in $(k-2)$-ary $\DynProp$, even when only insertions are allowed. We believe that finding simple queries that cannot be maintained will advance the understanding of dynamic complexity. Using the same proof technique as is used for the main result we also exhibit a $\QFO[\exists^*\forall^*]$-definable query that cannot be maintained in $\DynProp$; this improves the result \mbox{from \cite{GeladeMS12}}. Second, the main result establishes the first arity hierarchy for graph queries, although for a weak fragment of $\DynFO$ and for insertions only.  

The proof of the lower bound uses upper and lower bounds for Ramsey numbers. This has been quite curious for us.

A natural question is how far this method to prove lower bounds can be pushed. As an intermediate step between the quantifier-free fragment and \DynFO itself, Hesse suggested the study of quantifier-free update programs with auxiliary functions  \cite{Hesse03}. The main result can be extended as follows.

\begin{itemize}
 \item [\textbf{Extension of the main result:}] $k$-clique ($k\geq 3$) cannot be maintained by a quantifier-free update program of arity $k-2$ with unary auxiliary functions.
\end{itemize}

So far there have been only two lower bounds for dynamic classes with auxiliary functions. Alternating reachability was actually shown to be not maintainable in the quantifier-free fragment of $\DynFO$ even in the presence of a successor and a predecessor function \cite{GeladeMS12}. Further, in \cite{ZeumeS15reach}, it was shown that reachability cannot be maintained in unary \DynProp with unary auxiliary functions. Thus our extension is a first lower bound for arbitrary unary auxiliary functions and $k$-ary auxiliary relations, for every fixed $k$. We also explain why the lower bound technique does not extend to binary auxiliary functions.  To this end we show that binary $\DynQF$ can maintain every boolean graph property when the domain is large with respect to the actually used domain.

A preliminary version of this work appeared in \cite{Zeume14}. It was without many of the proofs and did not contain the lower bound for a $\QFO[\exists^*\forall^*]$-definable query.

\paragraph{Related work} Up to now we mentioned only work immediately relevant for this work. For the interested reader we give a short list of further related work.

Further lower bounds have been shown in \cite{DongLW95, DongLW03, GraedelS12}. Further upper bounds have been shown in \cite{Etessami98, Hesse03reach, WeberS07, GraedelS12}. Many other aspects such as whether the auxiliary relations are determined by the current structure (see e.g.\ \cite{PatnaikI97, DongS97, GraedelS12})  and the presence of an order (see e.g.\ \cite{GraedelS12})  have been studied.

\paragraph{Outline} In Section \ref{section:preliminaries} we fix some of our notations and in Section \ref{section:setting} we recapitulate the formal dynamic complexity framework. In Sections \ref{section:positiveefo} and \ref{section:aritylowerbound} we prove the upper and lower bound of the main result, respectively. In Section \ref{section:auxfunctions} we study the extension of $\DynProp$ by auxiliary functions. 

\paragraph{Acknowledgement} I am grateful to Thomas Schwentick for encouraging discussions and many suggestions for improving a draft of this work. Further I thank Samir Datta for fruitful discussions while he was visiting Dortmund. I thank Nils Vortmeier for proofreading.

  \section{Preliminaries}\label{section:preliminaries}
    We fix some of our notations. Most notations are reused from \cite{ZeumeS15reach}. The reader can feel free to skip this section and return when encountering unknown notations. 

A \textit{domain} $\domain$ is a finite set. A (relational) \emph{schema} $\schema$ consists of a set $\relSchema$ of relation symbols  and a set $\conSchema$ of constant symbols together with an arity function $\arity: \relSchema \rightarrow \N$. A \emph{database} $\db$ of schema $\schema$ with domain $\domain$ is a mapping that assigns to every relation symbol $R \in \relSchema$ a relation of arity $\arity(R)$ over $\domain$  and to every constant symbol $c \in \conSchema$ an element (called \textit{constant}) \mbox{from $\domain$}.

A  $\schema$-\emph{structure} $\struc$ is a pair $(\domain, \db)$ where $\db$ is a database with schema $\schema$ and domain $\domain$. If $\struc$ is a structure over domain $\domain$ and $\domain'$ is a subset of $\domain$ that contains all constants of $\struc$, then the substructure of $\struc$ induced by $\domain'$ is \mbox{denoted by $\restrict{\struc}{D'}$.}

A tuple $\vec a = (a_1, \ldots, a_k)$ is \emph{$\prec$-ordered} with respect to a linear order\footnote{All linear orders in this work are strict.}  $\prec$ of the domain, if $a_1 \norder \ldots \norder a_k$. The \emph{$k$-ary atomic type}  \type{\struc}{\vec a} of $\vec a$ over $\domain$ with respect to $\struc$ is the set of all atomic formulas $\varphi(\vec x)$ with \mbox{$\vec x = (x_1, \ldots, x_k)$} for which $\varphi(\vec a)$ holds in $\struc$, where $\varphi(\vec a)$ is short for the substitution of $\vec x$ by $\vec a$ in $\varphi$. As we only consider atomic types here, we will often simply say type instead of atomic type.

For a set $A$, denote by $A^k$ the set of all $k$-tuples over $A$ and, following \cite{GrahamRS1990}, by $\allsubsets{A}{k}$ the set of all $k$-element subsets of $A$. 
A \textit{$k$-hypergraph} $G$ is a pair $(V, E)$ where $V$ is a set and $E$ is a subset \mbox{of $\allsubsets{V}{k}$}. If $E = \allsubsets{V}{k}$  then $G$ is called \textit{complete}. 
An \textit{$r$-coloring} $\col$ of $G$ is a mapping that assigns to every edge in $E$ a color from $\{1, \ldots, r\}$. A \textit{$r$-colored $k$-hypergraph} is a pair $(G, \col)$ where $G$ is a $k$-hypergraph and $col$ is a $r$-coloring of $G$. If the name of the $r$-coloring is not important we also say \textit{$G$ is $r$-colored}.

A (directed) graph $G = (V, E)$ is in $\clique{k}$ if $V$ contains $k$ nodes $v_1,\ldots, v_k$ such that $(v_i, v_j) \in E$ or $(v_j, v_i) \in E$ for all $1 \leq i, j \leq k$.
 
  \section{Dynamic Setting}\label{section:setting}
    The following introduction to dynamic descriptive complexity is borrowed from previous \mbox{work \cite{ZeumeS15reach, ZeumeS14dyncqwithremark}}. Although the focus of this work is on maintaining the $k$-clique query under insertions, we introduce the general dynamic complexity framework in order to be able to give a broader discussion of concrete results.

A \emph{dynamic instance}\defindex{dynamic instance} of a query $\query$ is a pair $(\db, \alpha)$, where $\db$ is a database over some finite domain $\domain$ and $\alpha$ is a sequence of modifications to~$\db$. Here, a \emph{modification}\defindex{modification} is either an insertion of a tuple over $\domain$ into a relation of~$\db$ or a deletion of a tuple from a relation of~$\db$. The result of $\query$ for $(\db, \alpha)$ is the relation that is obtained by first applying the modifications \mbox{from $\alpha$} to $\db$ and then evaluating $\query$ on the resulting  database. We use the Greek letters $\alpha$ and $\beta$ to denote modifications as well as modification sequences.
The database resulting from applying a modification $\alpha$ to a database $\db$ is denoted by $\alpha(\db)$. The result $\updateDB{\alpha}{\db}$ of applying a sequence of modifications $\alpha \df \alpha_1 \ldots \alpha_m$ to a database $\db$ is defined by~$\updateDB{\alpha}{\db} \df \updateDB{\alpha_m}{\ldots (\updateDB{\alpha_1}{\db})\ldots}$.

Dynamic programs, to be defined next, consist of an initialization mechanism and an update program.  The former  yields, for every (input) database $\db$,  an initial state with initial auxiliary  data. The latter defines the new state of the dynamic program for each possible modification.

A \emph{dynamic schema} is a tuple  $(\inpSchema, \auxSchema)$ where $\inpSchema$ and $\auxSchema$ are the schemas of the input database and the auxiliary database, respectively. For the moment schema $\auxSchema$ may not contain constant symbols. This will be adapted in Section~\ref{section:auxfunctions}.  We always let $\tau\df\inpSchema\cup\auxSchema$. 

\begin{definition}(Update program)\label{def:updateprog}
  An \emph{update program} \prog over a dynamic schema \mbox{$(\inpSchema, \auxSchema)$} 
  is a set of first-order formulas (called \textit{update formulas} in the following) that contains,  for every relation symbol $R$ in $\auxSchema$ and every
  $\delta \in \{\ins_S, \del_S\}$ where $S$ is a relation symbol from $\inpSchema$, an update formula  $\uf{R}{\delta}{\vec x}{\vec y}$ over the schema $\schema$  where $\vec x$ and $\vec y$ have the same arity as $S$ and $R$, respectively.
\end{definition}

A \emph{program state} $\state$ over dynamic schema \mbox{$(\inpSchema, \auxSchema)$} is a structure $(\domain, \inp,  \aux)$ where $D$ is a finite domain, $\inp$ is a database over the input schema (the \emph{current database}) and $\aux$ is a database over the auxiliary schema (the \emph{auxiliary database}).

The \emph{semantics of update programs} is as follows\defindex{absolute semantics}. Let $P$ be an update program, $\state=(\domain, \inp,\aux)$ be a program state and $\alpha = \delta(\vec a)$ a modification where $\vec a$ is a tuple over $\domain$ and \mbox{$\delta \in \{\ins_S, \del_S\}$} for some $S \in \inpSchema$. If $P$ is in state $\state$ then the application of $\alpha$ yields the new state $\updateState{\prog}{\alpha}{\state} \df (\domain, \alpha(\inp), \aux')$  where, in $\aux'$, a relation symbol $R \in \auxSchema$ is interpreted by $\{\vec b \mid \state \models \uf{R}{\delta}{\vec a}{\vec b}\}$. The effect $P_\alpha(\state)$\defindex{P@$P_\alpha(\state)$} of applying a modification sequence $\alpha \df \alpha_1 \ldots \alpha_m$ to a state $\state$ is the \mbox{state $\updateState{P}{\alpha_m}{\ldots (\updateState{P}{\alpha_1}{\state})\ldots}$}.

\begin{definition}(Dynamic program) \label{definition:dynprog}
  A \emph{dynamic program} is a triple $(P,\init,Q)$, where
  \begin{compactitem}
   \item  $P$ is an update program over some dynamic schema
  \mbox{$(\inpSchema, \auxSchema)$}, 
    \item \init is a mapping that maps $\inpSchema$-databases to $\auxSchema$-databases, and 
    \item $Q\in\auxSchema$ is a designated \emph{query symbol}.
  \end{compactitem}
\end{definition}

A dynamic program $\calP=(P,\init,Q)$ \emph{maintains}  a dynamic query  $\dynProb{$\query$}$ if, for every dynamic instance $(\db,\alpha)$, the relation $\query(\alpha(\db))$ coincides with the query relation $Q^\state$ in the state \mbox{$\state=P_\alpha(\state_\init(\db))$}, where $\state_\init(\db)$ is the initial state for $\db$, i.e.\ \mbox{$\state_\init(\db) \df (\domain, \db,  \auxInit(\db))$}.

\begin{definition}(\DynFO and \DynProp) \label{definition:dync}
  \DynFO is the class of all dynamic queries that  can be maintained by dynamic programs with first-order update formulas and arbitrary initialization mappings. $\DynProp$ is the subclass of $\DynFO$, where update formulas are not allowed to use quantifiers.
  A dynamic program is \emph{$k$-ary} if the arity
  of its auxiliary relation symbols is at most $k$. By $k$-ary
  $\DynProp$ (resp. $\DynFO$) we refer to dynamic queries that can be
  maintained with $k$-ary dynamic programs. 
\end{definition}

In the literature, classes with restricted initialization mappings have been studied as well, \mbox{see \cite{ZeumeS15reach}} for a discussion. The choice made here is not a real restriction as lower bounds proved for arbitrary initialization hold for restricted initialization as well. On the other hand, our upper bounds also hold for other settings of initialization; with the single exception of Theorem \ref{theorem:padding}, which requires arbitrary initialization. Furthermore our results also hold in the related setting where domains can be infinite.

  \section{$k$-Clique Can Be Maintained under Insertions with \mbox{Arity $k-1$}}\label{section:positiveefo}
    In this section we prove that the $k$-clique query can be maintained in $(k-1)$-ary $\DynProp$ when only edge insertions are allowed. Instead of proving this result directly, we show that the class of all semi-positive existential first-order queries can be maintained in $\DynProp$ under insertions.

A \emph{positive existential first-order query} over schema $\schema$ is a query that can be expressed by a first-order formula of the form $\varphi(\vec y) = \exists \vec x \psi(\vec x, \vec y)$ where $\psi$ is a quantifier-free formula that contains no negations. \emph{Semi-positive existential first-order queries} may contain literals of the form $z_i \neq z_j$.

We will prove that every semi-positive existential first-order query can be maintained in $\DynProp$ when only insertions are allowed. More precisely, it will be shown that $(k-1)$-ary \DynProp is sufficient for boolean queries with $k$ existential quantifiers. In particular \clique{k} can be maintained in $(k-1)$-ary $\DynProp$. Before turning to the proof we give some intuition.
\begin{example}\label{example:cliqueupperbound}
  We show how to maintain \clique{3} in binary $\DynProp$ under insertions. The very simple idea is to use an additional binary auxiliary \mbox{relation $R$} that stores all edges whose insertion would complete a 3-clique. Hence a tuple $(a_1,a_2)$ is inserted into $R$ as soon as deciding whether there is a $3$-clique containing the nodes $a_1$ and $a_2$ only depends on those two nodes. We refer to Figure \ref{figure:cliqueupperbound} for an illustration.

  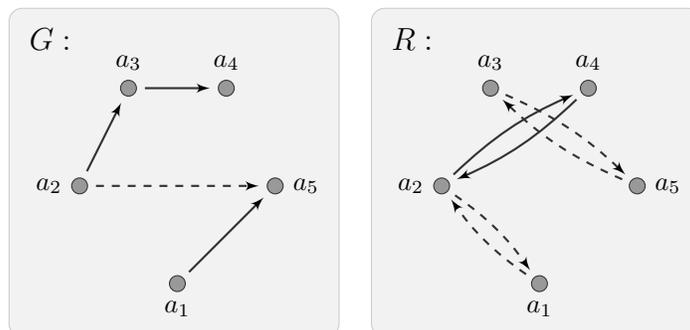
\begin{figure}[t] 
    \begin{center}
     \scalebox{1.0}{
	      \begin{tikzpicture}[
		xscale=1.3,
		yscale=1.3,
		show background rectangle
			      ]
	      
	    \node (tmp) at (-0.3,2) [nameNode,font=\large] {$G:$};
	      
	    \node (1) at (0, 0.5)[mnode, label={left:$a_2$}] {};
	    \node (2) at (0.5, 1.5)[mnode, label={above:$a_3$}] {};
	    \node (3) at (1.5, 1.5)[mnode, label={above:$a_4$}] {};
	    \node (4) at (2, 0.5)[mnode, label={right:$a_5$}] {};
	    \node (5) at (1, -0.5)[mnode, label={below:$a_1$}] {};
	    	    	    
	    \draw [dEdge] (1) to (2);
	    \draw [dEdge] (2) to (3);
	    \draw [dEdge] (5) to (4);
	    \draw [dDashedEdge] (1) to (4);

      \end{tikzpicture}\hspace{3mm}
	      \begin{tikzpicture}[
		xscale=1.3,
		yscale=1.3,
		show background rectangle
			      ]

		\node (tmp) at (-0.3,2) [nameNode,font=\large] {$R:$};  
		\node (1a) at (0, 0.5)[mnode, label={left:$a_2$}] {};
		\node (2a) at (0.5, 1.5)[mnode, label={above:$a_3$}] {};
		\node (3a) at (1.5, 1.5)[mnode, label={above:$a_4$}] {};
		\node (4a) at (2, 0.5)[mnode, label={right:$a_5$}] {};
		\node (5a) at (1, -0.5)[mnode, label={below:$a_1$}] {};

	    \draw  [dEdge, bend left=10](1a) to (3a);
	    \draw  [dEdge, bend left=10](3a) to (1a);
	    \draw  [dDashedEdge, bend left=10](1a) to (5a);
	    \draw  [dDashedEdge, bend left=10](5a) to (1a);
	    \draw  [dDashedEdge, bend left=10](2a) to (4a);
	    \draw  [dDashedEdge, bend left=10](4a) to (2a);
                                                            	    
	    \end{tikzpicture}
	  }
      \caption{Illustration of the construction from Example \ref{example:cliqueupperbound}. Inserting the edge $(a_2, a_5)$ into $G$ leads, e.g., to the insertion of $(a_1, a_2)$ into $R$ since inserting $(a_1, a_2)$ into $G$ would now complete a 3-clique. The tuple $(a_1, a_2)$ is inserted into $R$ by the dynamic program since chosing $(u, v, x, y)$ as $(a_2, a_5, a_1, a_2)$ satisfies the update formula $\uf{R}{\ins E}{u,v}{x,y}$.}
      \label{figure:cliqueupperbound}
    \end{center}\vspace{-7mm}
  \end{figure}

  Thus the update formula for $R$ is
    \begin{multline*}
     \uf{R}{\ins E}{u,v}{x,y} \df u \neq v \wedge x \neq y \wedge \Big(\big(E\{u,y\} \wedge v = x\big) \vee \big(E\{u,x\} \wedge v = y\big)\\ \vee \big(E\{v,y\} \wedge u = x\big) \vee \big(E\{v,x\} \wedge u = y\big)\Big)
    \end{multline*}
  where $E\{x,y\}$ is an abbreviation for $E(x,y) \vee E(y,x)$.
  
  The update formula for the query symbol $Q$ is 
    $\uf{Q}{\ins E}{u,v}{x,y} = Q \vee R(u,v)$.    \qed
\end{example}

The general proof for arbitrary semi-positive existential first-order properties extends the approach from the previous example. 

\begin{theorem}\label{theorem:positiveefo}
  An $\ell$-ary query expressible by a semi-positive existential first-order formula with $k$ quantifiers can be maintained under insertions in $(\ell+k-1)$-ary $\DynProp$.
\end{theorem}

\begin{proof} 
  For simplicity we restrict the proof to boolean graph queries. The proof easily carries over to arbitrary semi-positive existential queries.

  We give the intuition first. Basically a semi-positive existential sentence with $k$ quantifiers can state which (not necessarily induced) subgraphs with $k$ nodes shall occur in a graph. Therefore it is sufficient to construct a dynamic quantifier-free program that maintains whether the input graph contains a subgraph $H$. Such a program can work as follows. For every induced, proper subgraph $H' = \{u_1, \ldots, u_m\}$ of $H$, the program maintains an auxiliary relation that stores all tuples $\vec a = (a_1, \ldots, a_m)$ such that inserting $H'$ into $\{a_1, \ldots, a_m\}$ (with $a_i$ corresponding to $u_i$) yields a graph that contains $H$. 
  
  In particular, auxiliary relations have arity at most $k-1$ (as only proper subgraphs of $H$ have a corresponding auxiliary relation). Furthermore the graph $H$ is contained in the input graph whenever the value of the $0$-ary relation corresponding to the empty subgraph of $H$ is true. In the example above, the relation $R$ is the relation for the subgraph of the 3-clique graph that consists of a single edge, and the designated query relation is the $0$-ary relation for the empty subgraph.
  
  Those auxiliary relations can be updated as follows. Assume that a tuple $\vec a = (a_1, \ldots, a_m)$ is contained in the relation corresponding to $H'$. If, after the insertion of an edge with end point $a_m$, every edge from $u_m$ in $H'$ has a corresponding edge from $a_m$ in the graph induced by $\{a_1, \ldots, a_m\}$, then the tuple $\pvec a' = (a_1, \ldots, a_{m-1})$ has to be inserted into the auxiliary relation for the induced subgraph $\restrict{H'}{\{u_1, \ldots, u_{m-1}\}}$. This is because inserting the graph  $\restrict{H'}{\{u_1, \ldots, u_{m-1}\}}$ into $\{a_1, \ldots, a_{m-1}\}$ will now yield a graph that \mbox{contains $H$}. Observe that for those updates no quantifiers are needed.

  In the following we make the intuitive idea outlined above more precise. We first show how a quantifier-free dynamic program can maintain whether the input graph contains a certain (not necessarily induced) subgraph. Afterwards we show how to combine the programs for several subgraphs in order to maintain an arbitrary semi-positive existential formula. 

  For the first step it will be technically easier not to speak about subgraphs $H'$ of $H$ (as in the intuition above) but to work with partitions \mbox{of $H$}. We introduce this notion as well as other useful notions next. Let $H$ be a graph. A tuple $(\vec y, \vec z)$ is called a \emph{partition} of $H$ if it contains every node of $H$ exactly once. The subgraph of $H$ induced by $\vec y$  is denoted by $\restrict{H}{\vec y}$; the graph obtained from $H$ by removing the edges of $\restrict{H}{\vec y}$ is denoted by $H_{(\vec y, \vec z)}$.

  Now let $G = (V,E)$ and $H = (V', E')$ be graphs, and let $(\vec y, \vec z)$ be an arbitrary partition of $H$ with $|\vec y| = \ell$. We say that an $\ell$-ary tuple $\vec a$ can be extended to $H_{(\vec y, \vec z)}$, if there is a $|\vec z|$-tuple $\vec b$ such that the mapping $\pi$ defined by \mbox{$\pi(\vec y, \vec z) \df (\vec a, \vec b)$} maps edges in $H_{(\vec y, \vec z)}$ to edges in $G$. Intuitively $\vec a$ can be extended to $H_{(\vec y, \vec z)}$ when deciding whether $H$ is a subgraph of $G$, where $\vec y$ corresponds to $\vec a$, depends only on $\vec a$ and not on nodes of $G$ not contained \mbox{in $\vec a$}. See Figure \ref{figure:positive_efo} for an illustration.

    \begin{figure}[t]
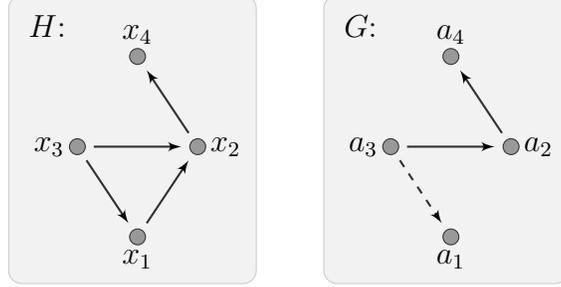
 
    \begin{center}
     \scalebox{1.0}{
        \picDynPropUpperBoundConstructionB
            }
            \hspace{3mm}
     \scalebox{1.0}{
        \picDynPropUpperBoundConstructionA
      }
      \caption{Illustration of the notions used in Theorem \ref{theorem:positiveefo}. The graph $H$ is the graph defined by the existential semi-positive formula $\exists x_1 \exists x_2 \exists x_3 \exists x_4 \big( \bigwedge_{i \neq j} x_i \neq x_j \wedge E(x_3, x_1) \wedge E(x_1, x_2) \wedge E(x_3, x_2) \wedge E(x_2, x_4)\big)$.  Before inserting the edge $(a_3, a_1)$ into $G$, the tuple $(a_1, a_2, a_3)$ can be extended to $H_{((x_1, x_2, x_3), x_4)}$, but $(a_1, a_2)$ does not extend to $H_{((x_1, x_2), (x_3, x_4))}$. After inserting the edge  $(a_3, a_1)$, the tuple $(a_1, a_2)$ can be extended to $H_{((x_1, x_2), (x_3, x_4))}$ as well.}
      \label{figure:positive_efo}
    \end{center}\vspace{-7mm}
  \end{figure}

  Let $\vec a = (a_1, \ldots, a_\ell)$ be a tuple that can be extended to $H_{(\vec y, \vec z)}$. Then a node $a_i$ is called \emph{saturated} with respect to a partition $(\vec y, \vec z)$ and $\vec a$ if $(a_i, a_j)$ (respectively $(a_j, a_i)$) is an edge in $G$ whenever $(y_i, y_j)$ (respectively $(y_j, y_i)$) is an edge in $H$. A tuple $(c,d)$ is \emph{critical} for $a_i$ with respect to a partition $(\vec y, \vec z)$ and $\vec a$ if $a_i$ is not saturated in $G$ but it is saturated in $G + (c, d)$. In Figure \ref{figure:positive_efo}, the tuple $(a_3, a_1)$ is critical for $a_3$ with respect to the partition $((x_1, x_2, x_3), x_4)$ and the tuple $(a_1, a_2)$. Observe that therefore the insertion of the edge $(a_3, a_1)$ yields a graph where $(a_1, a_2)$ can be extended \mbox{to $H_{((x_1, x_2), (x_3, x_4))}$.}
  
  We are now ready to construct a $\DynProp$-program $\prog$ that maintains whether the input graph contains a graph $H$ as (not necessarily induced) subgraph. The program $\prog$ has an auxiliary relation $R_{(\vec y, \vec z)}$ of arity $|\vec y|$ for every partition $(\vec y, \vec z)$ of $H$ with $|\vec z| \geq 1$.  The intention is that, for a state $\calS$ with input graph $G$, a tuple $\vec a$ is in $R^\state_{(\vec y, \vec z)}$ whenever $\vec a$ extends to $H_{(\vec y, \vec z)}$ in $G$. Thus $R_{(\vec y, \vec z)}$ corresponds to the auxiliary relation for $\restrict{H}{\vec y}$ in the intuitive explanation above. The condition $|\vec z| \geq 1$ ensures that the auxiliary relations are of arity at most $|H|-1$.
   
  Before sketching the construction of the update formulas it is illustrative to see what happens when inserting the edge $(a_3, a_1)$ in Figure \ref{figure:positive_efo}.  We observed above that this yields a graph where $(a_1, a_2)$ can be extended to $H_{((x_1, x_2), (x_3, x_4))}$. Therefore $(a_1, a_2)$ should be inserted into the auxiliary relation $R_{((x_1, x_2), (x_3, x_4))}$. However, this update of  $R_{((x_1, x_2), (x_3, x_4))}$ can be made without quantifiers since it is sufficient to verify that $(a_1, a_2, a_3)$ is already in $R_{((x_1, x_2, x_3), x_4)}$ and that $(a_1, a_3)$ was critical. This involves the nodes $a_1$, $a_2$ and $a_3$ only.
  
  In general, when an edge $e$ is inserted, the update formulas of $\prog$ check for which nodes and partitions the edge is critical; and adapt the auxiliary relations accordingly.
  
  For updating a relation $R_{(\vec y, \vec z)}$ with $\vec y = (y_1, \ldots, y_\ell)$ and $\vec z = (z_1, \ldots, z_{k-\ell})$ the update formula $\uf{R}{\ins\;E}{u,v}{\vec y}$ has to check whether there is some $R_{(\pvec y', \pvec z')}$ with $\pvec y' = (y_1, \ldots, y_i, z_j, y_{i+1} \ldots, y_\ell)$ and $\pvec z' = (z_1, \ldots, z_{j-1}, z_{j+1}, \ldots, z_{k-\ell})$ such that the insertion of $(u,v)$ saturates $z_j$. It is also possible that the insertion of a single edge saturates two nodes, this case is very similar and will not be treated in detail here.
   
  The formula $\uf{R}{\ins\;E}{u,v}{\vec y}$ is a conjunction of formulas $\varphi_u$, $\varphi_v$ and  $\varphi_{u,v}$ responsible for dealing with the cases where $u$, $v$ and both $u$ and $v$ are being saturated. We only exhibit $\varphi_u$:
   
   \begin{align*}
             \varphi_u \df & \bigvee_{\substack{\text{For all }(\pvec y', \pvec z') \text{ with}\\ \pvec y' = (y_1, \ldots, y_i, z_j, y_{i+1}, \ldots, y_\ell) \\  \pvec z' = (z_1, \ldots, z_{j-1}, z_{j+1}, \ldots, z_{k-\ell})}} \Big( & R_{(\pvec y', \pvec z')}(y_1, \ldots, y_i, u, y_{i+1}, \ldots, y_\ell) \wedge \bigwedge_{i'} u \neq y_{i'} \\
      & & \wedge \bigwedge_{(z_j, y_{i'}) \in \restrict{H}{\pvec y'}} E(u, y_{i'}) \wedge \bigwedge_{(y_{i'}, z_j) \in \restrict{H}{\pvec y'}} E(y_{i'}, u) \Big)
   \end{align*}
   The other formulas are very similar. This completes the construction of $\prog$.

   It remains to construct a quantifier-free dynamic program for an arbitrary semi-positive existential formula using quantifier-free programs for subgraphs. To this end let $\varphi = \exists \vec x \psi(\vec x)$ be an arbitrary semi-positive existential first-order formula. We show how to translate $\varphi$ into an equivalent disjunction of formulas $\varphi_i$ of the form
    $$\varphi_i = \exists \vec x_i \bigwedge_{y, y' \in \vec x_i} \Big( y \neq y' \wedge \psi_i(\vec x_i) \Big)$$
   where each $\psi_i$ is a conjunction of atoms over $\{E\}$ and $|\vec x_i| \leq |\vec x|$. 

   Observe that the quantifier-free part of each $\varphi_i$ encodes a subgraph $H_i$. Hence a graph $G$ satisfies $\varphi$ if and only if one of the graphs $H_i$ is a subgraph of~$G$. Thus a program maintaining the query defined by $\varphi$ can be constructed by combining the dynamic programs for all $H_i$ in a straightforward way.

   We now sketch how to translate $\varphi$ into the form stated above. First $\varphi$ is rewritten as disjunction of conjunctive queries, that is as $\bigvee_i \exists \vec y_i \gamma_i(\vec y_i)$ where each $\gamma_i$ is a conjunction of positive literals and literals of the form $x \neq x'$. Afterwards each $\exists \vec y_i \gamma_i(\vec y_i)$ is rewritten into an equivalent disjunction over all equality types on the variables in $\vec y_i$, that is as
      $$\bigvee_\varepsilon \exists \vec y_{i, \varepsilon} \Big(\bigwedge_{y,y' \in \vec y_{i, \varepsilon}} y \neq y' \wedge \varphi_{i, \varepsilon}(\vec y_{i, \varepsilon})\Big)$$
    where $\varepsilon$ is over all equality types and  $\varphi_{i, \varepsilon}$ is a conjunction of atoms \mbox{over $\{E\}$}.
\end{proof}
 
    \section{\hspace{-2mm}$k$-Clique \hspace{-0.2mm}Cannot \hspace{-0.2mm}Be \hspace{-0.2mm}Maintained \hspace{-0.2mm}with \hspace{-0.2mm}Arity \hspace{-0.2mm}$k-2$}\label{section:aritylowerbound}
    In this section we prove that the $k$-clique query cannot be maintained by a $(k-2)$-ary quantifier-free update program when $k \geq 3$. The proof uses two main ingredients; the Substructure Lemma from \cite{GeladeMS12, ZeumeS15reach} and a new Ramsey-like lemma. We state those lemmas next. Towards the end of this section we apply the lower bound technique presented in this section to show that there is a first-order property expressible by a formula with only one quantifier alternation which cannot be maintained in $\DynProp$ (with arbitrary arity).

For the convenience of the reader we recall the intuition for the Substructure Lemma as presented in \cite{ZeumeS15reach}. When updating an auxiliary tuple $\vec c$ after an insertion or deletion of a tuple $\vec d$, a quantifier-free update formula has access to $\vec c$, $\vec d$, and the constants only. Thus, if a sequence of modifications changes only tuples from a substructure $\calA$ of $\calS$, then the auxiliary data of $\calA$ is not affected by information outside $\calA$. In particular, two isomorphic substructures $\calA$ and $\calB$ remain isomorphic, when corresponding modifications are applied to them.

For stating the Substructure Lemma we need the following notion of corresponding modifications in isomorphic structures. Let $\pi$ be an isomorphism from a structure $\calA$ to a structure $\calB$. Two modifications $\alpha=\delta(\vec a)$ on $\calA$ and $\alpha=\delta'(\vec b)$ on $\calB$ where $\delta, \delta' \in \{\ins_R, \del_R\}$ for some $R \in \inpSchema$ are said to be \textit{$\pi$-respecting} if $\delta = \delta'$ and $\vec b = \pi(\vec a)$. Two sequences $\alpha_1\cdots\alpha_m$ and $\beta'_1\cdots\beta'_m$ of modifications respect $\pi$ if $\alpha_i$ and $\alpha'_i$ are $\pi$-respecting for every $i\le m$.  Recall that $\updateState{P}{\alpha}{\calS}$ denotes the state obtained by executing the dynamic program $\prog$ for the modification sequence $\alpha$ from state $\calS$. 

The Substructure Lemma stated next is illustrated in Figure \ref{figure:lemma:substructure}.
\begin{figure}[t]
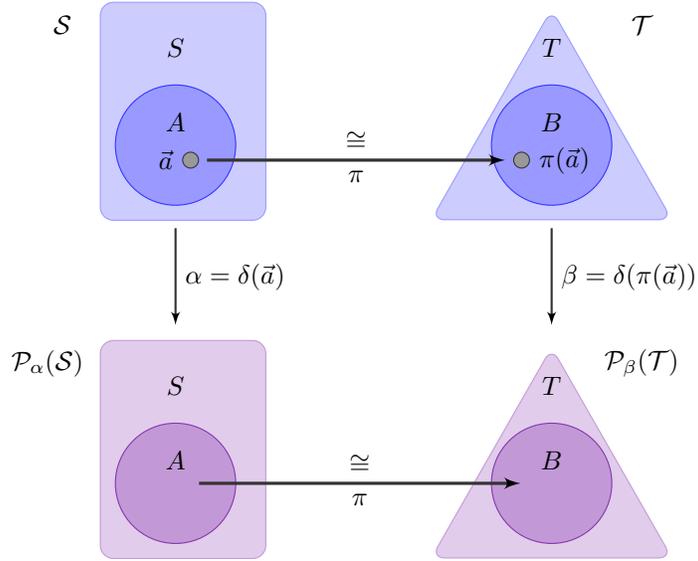
 
    \begin{center}
    \scalebox{1.0}{
     \picsubstructure
    }
    \caption{The statement of the Substructure Lemma. \label{figure:lemma:substructure}}
    \end{center}  \end{figure}
\begin{lemma}[Substructure Lemma \cite{GeladeMS12, ZeumeS15reach}]\label{lemma:substruccor}
  Let $\prog$ be a  \DynProp-program and let $\calS$ and $\calT$ be states of $\prog$ with domains $S$ and $T$. Further let $A \subseteq S$ and $B \subseteq T$ such that $\restrict{\calS}{A}$ and $\restrict{\calT}{B}$ are isomorphic via $\pi$. Then $\restrict{\updateState{P}{\alpha}{\calS}}{A}$ and $\restrict{\updateState{P}{\beta}{\calT}}{B}$ are isomorphic via  $\pi$ for all $\pi$-respecting modification sequences $\alpha$, $\beta$ on $A$ and $B$. In particular, if the query relation of $\prog$ is boolean, then it has the same value in $\updateState{P}{\alpha}{\calS}$ and $\updateState{P}{\beta}{\calT}$
\end{lemma}

The second ingredient exhibits a disparity between upper bounds for Ramsey numbers in $k$-ary structures and lower bounds for Ramsey numbers in \mbox{$(k+1)$}-dimensional hypergraphs. While the first condition in the following lemma guarantees the existence of a Ramsey clique of size $f(|A|)$ in $k$-ary structures \mbox{over $A$}, the second condition states that there is a 2-coloring of the complete $(k+1)$-hypergraph over $A$ that does not contain a Ramsey clique of size $f(|A|)$. This disparity is the key to the lower bound proof. 
\newcommand{\lemmaramseyfact}{
   Let $k \in \N$ be arbitrary and $\schema$ a $k$-ary schema. Then there is a function $f: \N \rightarrow \N$ and an $n \in \N$ such that for every domain $A$ larger than $n$ the following conditions are satisfied:
  \begin{enumerate}
  \item[(S1)] For every $\schema$-structure $\state$ over $A$ and every linear order $\prec$ on $A$ there is a subset $A'$ of $A$ of size $|A'| \geq f(|A|)$ such that all $\prec$-ordered $k$-tuples over $A'$ have the same type in $\state$.
  \item[(S2)] The set $\allsubsets{A}{k+1}$ of all $(k+1)$-hyperedges over $A$ can be partitioned into two sets $B$ and $B'$ such that for every set $A' \subseteq A$ of size $|A'| \geq f(|A|)$ there are $(k+1)$-hyperedges $b, b' \subseteq A'$ with $b \in B$ and $b' \in B'$.
  \end{enumerate}
}

\begin{lemma}\label{lemma:ramseyfact}
  \lemmaramseyfact
\end{lemma}

The two lemmas above can be used to obtain the lower bound for the $k$-clique query as follows. The proof of Lemma \ref{lemma:ramseyfact} will be presented afterwards.
\begin{theorem} \label{theorem:DynProp_lower_bound}
  \clique{(k+2)} ($k \geq 1$) cannot be maintained under insertions by a $k$-ary $\DynProp$-program.
\end{theorem}
\begin{proof}
  Towards a contradiction assume that there is a $k$-ary $\DynProp$-pro\mbox{gram $\prog$} over schema $\schema$ that maintains $\clique{(k+2)}$. Let $n$ and $f$ be as in Lemma \ref{lemma:ramseyfact}. For a set $A$ larger than $n$ let $\prec$ be an arbitrary order on $A$ and let $\domain \df A \disjointunion C$ be a domain with  $C \df \allsubsets{A}{k+1}$. Further let $B, B'$ be the partition of $\allsubsets{A}{k+1}$ guaranteed to exist by (S2) in Lemma \ref{lemma:ramseyfact}. 
  
  We consider a state $\state$ over domain $\domain$ where the input graph $G$ contains the following edges: 
  $$\{(b, {b_1}), (b, {b_2}), \ldots, (b, {b_{k+1}}) \mid  b = \{b_1, b_2,\ldots, b_{k+1}\} \in B\}$$
  See Figure \ref{figure:DynProp_lower_bound} for an illustration.
  \begin{figure}[t]
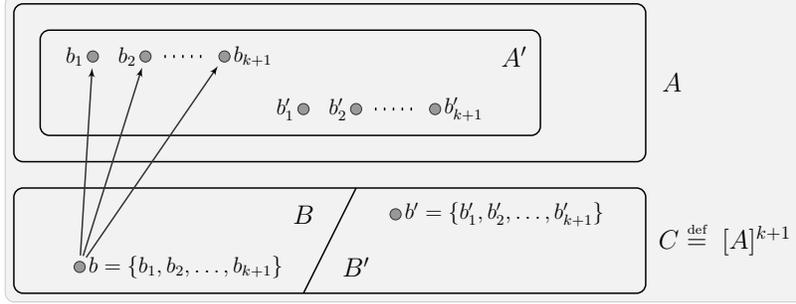
 
    \begin{center}
     \scalebox{0.7}{
        \picDynPropLowerBoundClique
      }
      \caption{The construction from the proof that $(k+2)$-clique cannot be maintained in $k$-ary $\DynProp$.}
      \label{figure:DynProp_lower_bound}
    \end{center}\vspace{-7mm}
  \end{figure}
  
  By Condition (S1) there is a subset $A'\subseteq A$  of size $|A'| \geq f(|\domain|)$ such that all ordered $k$-tuples over $A'$ have the same $\schema$-type in $\state$. Then by (S2) there are $(k+1)$-hyperedges \mbox{$b, b' \subseteq A'$} with $b \in B$ and $b' \in B'$. Without loss of generality $b = \{b_1, b_2, \ldots,  b_{k+1}\}$ with $b_1 \prec \ldots \prec b_{k+1}$ and $b' = \{b'_1, b'_2, \ldots, b'_{k+1}\}$ with $b'_1 \prec \ldots \prec b'_{k+1}$. By construction of the graph $G$, all elements in $b$ are connected to the node ${b} \in C$ while there is no node in $C$ connected to all elements of $b'$. Thus applying the modification sequences
  \begin{itemize}
    \item [($\alpha$)] Insert the edges $(b_i, b_j)$ in lexicographic order with respect to $\prec$.
    \item [($\beta$)] Insert the edges $(b'_i, b'_j)$ in lexicographic order with respect to $\prec$.
  \end{itemize}
  yields one graph with a $(k+2)$-clique and one graph without a $(k+2)$-clique, respectively. However, by the Substructure Lemma, the program $\prog$ yields the same result since the substructures induced by $\vec b = (b_1, \ldots, b_{k+1})$ and \mbox{$\pvec b' = (b'_1, \ldots, b'_{k+1})$} are isomorphic. This is the desired contradiction.
\end{proof}

In the following we prove Lemma \ref{lemma:ramseyfact}. The \emph{$k$-dimensional Ramsey} number for $r$ colors and clique-size $l$, denoted by $\ramsey{k}{l}{r}$, is the smallest \mbox{number $n$} such that every $r$-coloring of a complete $k$-hypergraph with $n$ nodes has a monochromatic clique of size $l$. The \emph{tower function} $\tower{k}{n}$ is defined by
$$\tower{k}{n} \df 2^{2^{.^{.^{.^{2^n}}}}}$$
with $(k-1)$ many $2$'s. The following classical result for asymptotic bounds on Ramsey numbers due to \Erdos, Hajnal and Rado is the key to prove Lemma~\ref{lemma:ramseyfact}. The concrete formulation is from \cite{DuffusLR95}.

 \begin{theorem}\label{theorem:ramseybounds}
    \cite{ErdosR52, ErdosHR65}
    Let $k$, $\ell$ and $r$ be positive integers. Then there are positive constants $c_k$, $c_{k,r}$ and $\ell_k$ such that
      \begin{enumerate}
        \item $\ramsey{k}{\ell}{r} \leq \tower{k}{c_{k,r}\ell}$
        \item $\ramsey{k}{\ell}{2} \geq \tower{k-1}{c_k\ell^2}$ for all $\ell \geq \ell_k$
      \end{enumerate}
  \end{theorem}
  The theorem immediately implies that (T1) Ramsey cliques in $r$-colored $k$-dimensional complete hypergraphs are of size at least  $\Omega(\klog{k-1}{n})$; and that (T2) there are 2-colorings of the $(k+1)$-dimensional complete hypergraphs such that monochromatic cliques are of size $O((\klog{k-1}{n})^{\frac{1}{2}})$. Here \mbox{$\klog{k}{n}$ denotes $\log(\log(\ldots (\log n) \ldots))$} with $k$ many $\log$'s.

  The conditions (T1) and (T2) are formalized and proved in the following corollary.
  
  \begin{corollary}\label{corollary:ramseybounds}
    Let $k$ and $r$ be integers. There are functions $g \in \Omega(\klog{k-1}{n})$ and $h \in O(\sqrt{\klog{k-2}{n}})$ such that:
    \begin{enumerate}
      \item Every $r$-colored complete $k$-hypergraph with $n$ nodes contains a monochromatic clique of size $g(n)$.
      \item The complete $k$-hypergraph with $n$ nodes can be 2-colored such that every monochromatic clique is of size at most $h(n)$.
    \end{enumerate}
  \end{corollary}
  \begin{proof}
    The corollary follows immediately from Theorem \ref{theorem:ramseybounds}. Define $g$ and $h$ by
    
    $$g(n) \df \Big\lfloor \tfrac{1}{c_{k,r}} \klog{k-1}{n}\Big\rfloor \text{ and } h(n) \df \Big\lceil \sqrt{\tfrac{1}{c_k}\klog{k-2}{n}} \Big\rceil+1$$
    
    where the constants $c_{k,r}$ and $c_k$ are as in Theorem \ref{theorem:ramseybounds}.
    
    For proving a), consider an arbitrary hypergraph $G$ with $n$ nodes, and an arbitrary $r$-coloring of $G$. Then, by Theorem \ref{theorem:ramseybounds}a), there is a monochromatic clique of size $\ell$ where $\ell$ is the maximal number such that \mbox{$\tower{k}{\ell c_{k,r}} \leq n$}. The number $\ell$ is exactly $g(n)$.
    
    For proving b), consider again an arbitrary hypergraph $G$ with $n$ nodes. By Theorem \ref{theorem:ramseybounds}b), there is a $2$-coloring without a  monochromatic clique of size $\ell$ where $\ell$ is the minimal number such that $n <  \tower{k-1}{\ell^2c_k}$.  Thus the largest monochromatic clique of $G$ is of size at most $h(n)$.
  \end{proof}

   The conditions (T1) and (T2) are already quite similar to the conditions (S1) and (S2). The major difference is that (T1) is about hypergraphs and not about structures with a $k$-ary schema. 
   
   Fortunately the upper bound from Theorem \ref{theorem:ramseybounds} can be generalized to Ramsey numbers for structures. To this end some notions need to be transferred from hypergraphs to structures. Let $\schema$ be a $k$-ary schema, let $\state$ be a $\schema$-structure over domain $\domain$ and let $\prec$ be a linear order on $\domain$. A subset $\domain' \subseteq \domain$ of the domain of $\state$ is called an \emph{$\prec$-ordered $\schema$-clique} if all $\prec$-ordered $k$-tuples $\vec a \in \domain'^{k}$ have the same $\schema$-type.  Recall that the type of a tuple $\vec a$ includes information of how $\vec a$ relates to the constants of the structure, and therefore all tuples over a $\schema$-clique relate in the same way to constants as well. Denote by $\ramsey{}{\ell}{\schema}$ the smallest number $n$ such that every $\schema$-structure with $n$ elements contains an $\prec$-ordered $\schema$-clique of size $\ell$, for every order $\prec$ of the domain.

\begin{theorem} \label{theorem:orderedramseybound}
  Let $\schema$ be a schema with maximal arity $k$ and let $\ell$ be a positive integer. Then there is a constant $c$ such that $\ramsey{}{\ell}{\schema} \leq \tower{k}{\ell c}$.
\end{theorem}
\begin{proof}
  The proof of Observation 1' in \cite[p. 11]{GeladeMS12} yields this bound. For the sake of completeness we repeat the full construction.
  
  Consider the schema $\schema$ and let $\Gamma$ be the set of all $k$-ary types for $\schema$. Let $\state$ be a $\schema$-structure over domain $\domain$ of size $\tower{k}{\ell c}$ where $c = |\Gamma|$. Further let $\prec$ be an arbitrary order on $\domain$. Define a coloring $\col$ of the complete $k$-dimensional hypergraph over domain $\domain$ with colors $\Gamma$ as follows.  An edge $\{e_1, \ldots, e_{k}\}$ with $e_1 \prec \ldots \prec e_{k}$ is colored by the type $\type{S}{e_1, \ldots, e_{k}}$. By Theorem \ref{theorem:ramseybounds} there is an induced monochromatic sub-$k$-hypergraph with domain $D' \subseteq D$ with $|D'| \geq \ell$. By the definition of the coloring $\col$, two $\prec$-ordered $k$-tuples over $D'$ have the same type and therefore $D'$ is a $\prec$-ordered $\schema$-clique in $\state$ as well.
\end{proof}

The previous theorem implies that Ramsey cliques in $k$-ary structures are of size at least  $\Omega(\klog{k-1}{n})$. The proof is analogous to the proof of \mbox{Corollary \ref{corollary:ramseybounds}}.

\begin{corollary}\label{corollary:orderedramseyboundb}
  Let $\schema$ be a schema with maximal arity $k$. There is a function \mbox{$g \in \Omega(\klog{k-1}{n})$} such that every $\schema$-structure with $n$ elements contains an ordered $\schema$-clique of size $g(n)$.
\end{corollary}
  
It remains to prove Lemma \ref{lemma:ramseyfact}. 
  
  \begin{proofof}{Lemma \ref{lemma:ramseyfact}}
    Let $k \in \N$ be arbitrary and let $\schema$ be a $k$-ary schema $\schema$. Choose $f \df g$ where $g \in \Omega(\klog{k-1}{n})$ is the function from Corollary \ref{corollary:orderedramseyboundb}. We show that there is an $n$ such that $f$ satisfies the conditions (S1) and (S2) for all domains larger than $n$.
    
    Let $h \in O(\sqrt{\klog{k-1}{n}})$ be the function guaranteed to exist for $k+1$ by Corollary \ref{corollary:ramseybounds}b).     Then $h \in o(f)$, and therefore there is an $n$ such that $f(n') > h(n')$ for all $n' > n$. Hence for every domain larger than $n$ condition (S1) is satisfied for $f$ due to Corollary \ref{corollary:orderedramseyboundb} and condition (S2) is satisfied due to \mbox{Corollary \ref{corollary:ramseybounds}}.
  \end{proofof}

Next we  apply the lower bound proof technique presented above in order to improve upon a result by Gelade et al. \cite{GeladeMS12}. They provided a lower bound for the alternating reachability problem. The use of very restricted graphs in the proof implies that there is a $\QFO[\exists^*\forall^*\exists^*]$-definable query that cannot be maintained in $\DynProp$.  We show that there is a first-order property expressible by a formula with only one quantifier alternation which cannot be maintained in $\DynProp$. It remains open whether there is a $\EFO$- or $\AFO$-property that is not maintainable in $\DynProp$.

\begin{theorem} \label{theorem:DynProp_lower_bound_EAFO}
  There is a $\QFO[\exists^*\forall^*]$-definable query which cannot be maintained by a $\DynProp$-program.
\end{theorem}
\begin{proof}
  Consider the graph schema $\{E\}$ extended by two constants $s$ and $t$. We show that the query $\query$ defined by $\varphi \df \exists x \forall y \big(E(s,x) \wedge (E(y,t) \rightarrow E(x,y))\big)$ cannot be maintained by any $\DynProp$-program. We remark that it is possible to remove the constants from the following construction by using more existential quantifiers.
  
  The proof is an adaption of the proof of Theorem \ref{theorem:DynProp_lower_bound}. Towards a contradiction assume that there is a $k$-ary $\DynProp$-pro\mbox{gram $\prog$} over $k$-ary schema $\schema$ that maintains $\query$. Let $n$, $f$, $A$, $C$, $B$, $B'$, $\prec$ be as in the proof of Theorem \ref{theorem:DynProp_lower_bound}.
  
  We consider a state $\state$ over domain $\domain$ where the input graph $G$ contains, as before, the edges  $$\{(b, {b_1}), (b, {b_2}), \ldots, (b, {b_{k+1}}) \mid  b = \{b_1, b_2,\ldots, b_{k+1}\} \in B\}$$
  and, additionally, the edges 
  $$\{(s,b) \mid b = \{b_1, b_2,\ldots, b_{k+1}\} \in B\}$$
  
  See Figure \ref{figure:DynProp_lower_bound_EAFO} for an illustration.
  \begin{figure}[t] 
    \begin{center}
     \scalebox{0.7}{
        \picDynPropLowerBoundEAFO
      }
      \caption{The construction from the proof that the $\QFO[\exists^*\forall^*]$-definable query \mbox{$\varphi \df \exists x \forall y \big(E(s,x) \wedge (E(y,t) \rightarrow E(x,y))\big)$} cannot be maintained in $\DynProp$. The edges inserted by the modification sequence $\alpha$ are dashed.}
      \label{figure:DynProp_lower_bound_EAFO}
    \end{center}\vspace{-7mm}
  \end{figure}
    
  By Lemma \ref{lemma:ramseyfact}, we can find tuples $b = \{b_1, b_2, \ldots,  b_{k+1}\}$ with $b_1 \prec \ldots \prec b_{k+1}$ and $b' = \{b'_1, b'_2, \ldots, b'_{k+1}\}$ with $b'_1 \prec \ldots \prec b'_{k+1}$ such that $(s, t, b)$ and $(s, t, b')$  have the same  $\schema$-type in $\state$. 
  
  However, applying the modification sequences
  \begin{itemize}
    \item [($\alpha$)] Insert the edges $(b_i, t)$ in lexicographic order with respect to $\prec$.
    \item [($\beta$)] Insert the edges $(b'_i, t)$ in lexicographic order with respect to $\prec$.
  \end{itemize}
  yields one graph that satisfies $\varphi$ and one graph that does not. However, by the Substructure Lemma, the program $\prog$ yields the same result. This is the desired contradiction.
\end{proof}
     
  \section{Adding Auxiliary Functions}\label{section:auxfunctions}

  In quantifier-free update programs, as considered up to here, only the modified and updated tuple as well as the constants can be accessed while updating an auxiliary tuple. Since lower bounds for first-order update programs where arbitrary elements can be accessed in updates seem to be out of reach for the moment, it seems natural to look for extensions of quantifier-free update programs that allow for accessing more elements in some restricted way.

With $\DynQF$, one such extension was proposed by Hesse. In addition to auxiliary relations, a $\DynQF$-program may maintain auxiliary functions. Those functions are updated by update terms that may use function symbols as well as an if-then-else construct.

While Hesse obtained upper bounds for $\DynQF$ only, in subsequent work some first lower bounds have been obtained. In \cite{GeladeMS12} it was shown that the alternating reachability problem cannot be maintained in $\DynProp$ extended by a fixed successor function and a fixed predecessor function. Later, in \cite{ZeumeS15reach}, the reachability problem was shown to be not maintainable in unary $\DynProp$ with additional (updatable) unary auxiliary functions.

In this section we continue the study of \DynProp extended by auxiliary functions. In the first part we prove that \clique{(k+2)} cannot be maintained  by a $k$-ary $\DynProp$-program with unary auxiliary functions, even if only insertions are considered. In the second part we discuss the expressiveness of  binary $\DynQF$ in large domains and argue why the lower bound technique for the $k$-clique query does not immediately translate. 

Before continuing, we repeat a toy example from \cite{ZeumeS14dyncqwithremark} which is designed to give an impression of $\DynQF$. For a more formal treatment we refer the reader to \cite{GeladeMS12} and \cite{ZeumeS15reach}.

\begin{example}
  Consider the unary graph query $\query(x)$ that returns all nodes $a$ of a given graph $G$ with maximal outdegree .

  We construct a unary $\DynQF$-program $\prog$ that maintains $\query$ in a unary relation denoted by the designated symbol $Q$.  The program uses two unary functions $\Succ$ and $\Pred$ that shall encode a successor and its corresponding predecessor relation on the domain. For simplicity, but without loss of generality, we therefore assume that the domain is of the form $\domain = \{0, \ldots, n-1\}$. For every state $\state$, the function $\Succ^\state$ is then the standard successor function \mbox{on $\domain$} (with $\Succ^\state(n-1) = n-1$), and $\Pred^\state$ is the standard predecessor function (with $\Pred^\state(0) = 0$). Both functions are initialized accordingly. 
  In the following, when we talk about a \emph{number}, we mean the element whose position in $\Succ$ is that number.
  The program has constants that represent the numbers $0$ and $1$.
  
  The program $\prog$ maintains two unary functions $\numEdges$ and $\numNodes$. The function $\numEdges$ counts, for every node $a$, the number of outgoing edges of $a$; more precisely $\numEdges(a) = b$ if and only if $b$ is the number of outgoing edges \mbox{of $a$}. The function $\numNodes$ counts, for every number $a$, the number of nodes with $a$ outgoing edges; more precisely $\numNodes(a) = b$ if and only if $b$ is the number of nodes with $a$ outgoing edges. A constant $\Max$ shall always point to the number $i$ such that $i$ is the maximal number of outgoing edges from some node in the current graph.
  
  When inserting an outgoing edge $(u,v)$ for a node $u$ that already has $a$ outgoing edges, the counter $\numEdges$ of $u$ is incremented from $a$ to $a+1$ and all other edge-counters remain unchanged. The counter $\numNodes$ of $a$ is decremented, the counter of $a+1$ is incremented, and all other node-counters remain unchanged. The number $\Max$ increases if, before the insertion, $u$ was a node with maximal number of outgoing edges. This yields the following update terms:
  \begin{align*}
    \ut{\numEdges}{\ins\;E}{u,v}{x} \df & \itewa\Big(\neg E(u,v) \wedge x=u,\Succ(\numEdges(x)),\numEdges(x)\Big) \\
    \ut{\numNodes}{\ins\;E}{u,v}{x} \df & \itewa\Big(\neg E(u,v) \wedge x = \numEdges(u), \Pred(\numNodes(x)), \\
      & \quad \itewa\big(\neg E(u,v) \wedge x = \Succ(\numEdges(u)), \Succ(\numNodes(x)), \\
      & \quad \quad \numNodes(x)  \big)\Big) \\
    \utw{\Max}{\ins\;E}{u,v} \df &  \itewa\Big(\Max = \numEdges(u) \wedge \neg E(u,v), \Succ(u),\Max \Big) 
  \end{align*}
  The $\ite{}{}{}$-construct chooses, depending on the predicate in its first argument, either the second or the third argument as result term.

  The update formula for the designated query symbol $Q$ is as follows:
   \begin{align*}
    \uf{Q}{\ins\;E}{u,v}{x} \df  \ut{\numEdges}{\ins\;E}{u,v}{x} = \utw{\Max}{\ins\;E}{u,v}
   \end{align*}
   
  The update terms and and the update formula for deletions are very similar.
    \qed
\end{example}

\subsection{Lower Bounds for Unary Functions}

In this section we generalize the lower bounds obtained so far as follows.

\begin{theorem}\label{theorem:unaryDynQF_lower_bound}
  \clique{(k+2)} ($k \geq 1$) cannot be maintained under insertions by a $k$-ary $\DynProp$-program with unary auxiliary functions.
\end{theorem}

\begin{theorem} \label{theorem:unaryDynQF_lower_bound_EAFO}
  There is a $\QFO[\exists^*\forall^*]$-definable query which cannot be maintained by a $\DynProp$-program with unary auxiliary functions.
\end{theorem}

The proofs are along the same lines as the proofs of \mbox{Theorem \ref{theorem:DynProp_lower_bound}} and Theorem \ref{theorem:DynProp_lower_bound_EAFO}. Instead of the Substructure Lemma for \DynProp a corresponding lemma for $\DynQF$ from \cite{GeladeMS12,ZeumeS15reach} is used. This Substructure Lemma is slightly more involved as it requires to exhibit isomorphic substructures that, additionally, have similar neighbourhoods. Before stating the Substructure Lemma for \DynQF and  proceeding with the proof we repeat some useful notions \mbox{from \cite{GeladeMS12,ZeumeS15reach}}.  

For the following definitions we fix two structures $\calS$ and $\calT$ with domains $S$ and $T$ over schema $\schema$. Here, and in the rest of this section, all auxiliary schemas are the disjoint union of a set $\relSchema$ of relation symbols and a set $\funSchema$ of function symbols. Denote by $\Terms{\schema}{m}$ the set of terms of nesting depth at most $m$ with function symbols from $\funSchema$. 

The \emph{$m$-neighborhood}  $\nb{A}{\state}{m}$ of a set $A \subseteq S$ is the set of all elements of $S$ that can be obtained by applying a term of nesting depth at most $m$ to a vector of elements from $A$. More precisely $\nb{A}{\state}{m}$ is the set 
\[
\{\sem{t}{(\state,\beta)}\mid t\in \Terms{\schema}{m}\text{ and } \beta(x)\in A, \text{for every variable $x$ in $t$}\}.
\]

While for the Substructure Lemma for $\DynProp$ it is sufficient to consider two isomorphic substructures, the Substructure Lemma for $\DynQF$ also takes their neighborhoods into account. The neighborhoods need to be similar in the following sense. Two subsets $A \subseteq S$, $B \subseteq T$ are \textit{$m$-similar}, if there is a  bijection \mbox{$\pi: \nb{A}{\calS}{m} \rightarrow \nb{B}{\calT}{m}$} such that
\begin{itemize}
 \item the restriction of $\pi$ to $A$ is a bijection of $A$ and $B$,
  \item $\pi$ satisfies the equation $\pi(t^\state(\vec a)) = t^\calT(\pi(\vec a))$ for all $t \in \Terms{\funSchema}{m}$ and all $\vec a$ over $A$, and
  \item $\pi$ preserves $\relSchema$ on $\nb{A}{\calS}{m}$.
\end{itemize}

We write $A\approx_m^{\pi,\calS,\calT} B$ to indicate that $A$ and $B$ are $m$-similar via $\pi$ in $\calS$ \mbox{and $\calT$}. 
Two tuples $(a_1,\ldots,a_p)$ and $(b_1,\ldots,b_p)$ are $m$-similar if $\{a_1,\ldots,a_p\}\approx_m^{\pi,\calS,\calT} \{b_1,\ldots,b_p\}$ via the isomorphism $\pi$ that maps $a_i$ to $b_i$, for every $i\in\{1,\ldots,p\}$. Note that if $A\approx_0^{\pi,\calS,\calT} B$, then $\restrict{\calS}{A}$ and $\restrict{\calT}{B}$ are $\relSchema$-isomorphic by the first and third property.

\begin{lemma}[Substructure lemma for \DynQF\cite{GeladeMS12, ZeumeS15reach}]\label{lemma:substruclemmafun}
  Let $\prog$ be a \DynQF program and let $\ell$ be some number. There is a number $m \in \N$ such that for all states $\calS$ and $\calT$ of $\prog$  with domains $S$ \mbox{and $T$}; and all subsets  $A$ and $B$ of $S$ and $T$, respectively, the following holds. If $A\approx_{m}^{\pi,\calS,\calT} B$, then $A\approx_0^{\pi,\updateState{P}{\alpha}{\calS},\updateState{P}{\beta}{\calT}} B$, for all $\pi$-respecting modification sequences $\alpha$ and $\beta$  on $A$ and $B$ of length at \mbox{most $\ell$}.
\end{lemma}

The lemma is slightly rephrased in comparison to \cite{ZeumeS15reach}. Here, the number $m$ is independent of the states $\calS$ and $\calT$. However, this follows immediately from the proof of the lemma in \cite{ZeumeS15reach}.
  
In order to apply the Substructure Lemma, it is necessary to find similar substructures. To this end the following analogon of Corollary \ref{corollary:orderedramseyboundb} for structures with unary functions can be used. Recall that \mbox{$\klog{k}{n}$ denotes $\log(\log(\ldots (\log n) \ldots))$} with $k$ many $\log$'s.
\begin{lemma}\label{lemma:unaryDynQF_large_similar}
  Let $\schema$ be a $k$-ary schema whose function symbols are of arity at most 1; and let $m \in \N$ be an arbitrary number. Then there is a function \mbox{$g \in \Omega(\klog{k-1}{n})$} such that for every $\schema$-structure $\state$ with domain $S$ and every linear order $\prec$ on $S$, there is a subset $S' \subseteq S$ of size $g(|S|)$ such that all $\prec$-ordered $k$-tuples over $S'$ are $m$-similar.
\end{lemma}
\begin{proof}
  The idea is to construct, from the structure $\calS$, a purely relational structure $\calT$ such that the type of a tuple $\vec a$ in $\calT$ encodes the type of the whole $m$-neighborhood of $\vec a$ in $\calS$. Then Corollary \ref{corollary:orderedramseyboundb} is applied to the structure $\calT$ in order to obtain $S'$.

  For the construction of $\calT$ we need some notions from the proof of \mbox{Theorem 5.4} in \cite{ZeumeS15reach}. Let $t_1, \ldots, t_\ell$ be the lexicographic enumeration of $\Terms{\schema}{m}$ with respect to some fixed order of the function symbols.   Let the \textit{$m$-neighborhood vector} $\nbv{c}{\state}{m}$ of an element $c$ in $\calS$ be the tuple $(c, t_1(c), \ldots, t_l(c))$. For a tuple $\vec c = (c_1, \ldots, c_p)$, the $m$-neighborhood vector $\nbv{\vec c}{\state}{m}$ of $\vec c$ is the tuple $(\nbv{c_1}{\state}{m}, \ldots, \nbv{c_p}{\state}{m})$.  

  The \emph{$m$-similarity type} of a $k$-ary tuple $\vec a$ is the (quantifier-free) $\schema$-type of the $m$-neighborhood tuple $\nbv{\vec a}{\calS}{m}$ of $\vec a$. Observe that for fixed $m$, $k$ and $\schema$ there are only finitely many similarity types. Denote the set of all such similarity types by $\Gamma$. Further observe that two tuples $\vec a$ and $\vec b$ with the same $m$-similarity type are $m$-similar. This is certified by the bijection that maps $\nbv{\vec a}{\calS}{m}$ to  $\nbv{\vec b}{\calS}{m}$ component-wise.

  For the construction of $\calT$ we assume, without loss of generality, that the schema of $\calS$ contains the equality symbol $=$. The structure $\calT$ is over the same domain as $\calS$ and uses the schema $\schema_\Gamma$ which contains a $k$-ary relation $R_\gamma$ for every $k$-ary similarity type $\gamma \in \Gamma$. A relation $R^\calT_\gamma$ contains all tuples $\vec a$ whose similarity type in $\calS$ is $\gamma$.
  
  Then, by Corollary \ref{corollary:orderedramseyboundb}, $\calT$ contains an $\prec$-ordered $\schema$-clique $S'$ of size $\Omega(\klog{k-1}{|S|})$.  We show that all $\prec$-ordered $k$-tuples over $S'$ are $m$-similar in the structure $\calS$. Therefore let $\vec a$ and $\vec b$ be two such tuples. By definition of $S'$ they have the same type in $\calT$ and therefore, by definition of $\calT$, their neighborhood vectors $\nbv{\vec a}{\calS}{m}$ and  $\nbv{\vec b}{\calS}{m}$ have the same type in $\calS$. Hence, by the observation from above, the tuples $\vec a$ and $\vec b$ are $m$-similar.
\end{proof}

\begin{proofof}{Theorem \ref{theorem:unaryDynQF_lower_bound} and Theorem \ref{theorem:unaryDynQF_lower_bound_EAFO}}
  The proofs are along the same lines as the proofs of \mbox{Theorem \ref{theorem:DynProp_lower_bound}} and Theorem \ref{theorem:DynProp_lower_bound_EAFO}. The only difference is that here we use Lemma \ref{lemma:unaryDynQF_large_similar} in order to obtain the set $A'$. The contradiction is then obtained by using the Substructure Lemma for $\DynQF$ and the modification sequences ($\alpha$) and ($\beta$) from \mbox{Theorem \ref{theorem:DynProp_lower_bound}} and Theorem \ref{theorem:DynProp_lower_bound_EAFO}, respectively.
\end{proofof}

  \subsection{Discussion of Binary Auxiliary Functions}

A natural question is whether the lower bounds transfer to $k$-ary auxiliary functions with $k \geq 2$. We conjecture that they do, but we will argue that the techniques used so far are not sufficient for proving lower bounds for binary auxiliary functions.

The fundamental difference between unary and binary auxiliary functions is that, on the one hand, unary functions can access elements that depend either on the tuple that has been modified in the input structure or on the auxiliary tuple under consideration but not on both. On the other hand binary functions can access elements that depend on both tuples.

A consequence is that binary $\DynQF$ can maintain every boolean graph property when the domain is large with respect to the actually used domain. We make this more precise. In the following we assume that all domains $\domain$ are a disjoint union of a modifiable domain $\domain^+$ and a non-modifiable \mbox{domain $\domain^-$}, and that modifications may only involve tuples over $\domain^+$. Auxiliary data, however, may use the full domain. A dynamic complexity class $\calC$ \emph{profits from padding} if every boolean graph property can be maintained whenever the non-modifiable domain is sufficiently large in comparison to the modifiable domain\footnote{Note that this type of padding differs from the padding technique used by Patnaik and Immerman for maintaining a \PTIME-complete problem in $\DynFO$ \cite{PatnaikI97}.}.

Above we have seen that $\DynProp$ with unary auxiliary functions does not profit from padding. 

\begin{theorem}\label{theorem:padding}
  Binary \DynQF profits from padding.
\end{theorem}
\begin{proof}
  First we show that ternary \DynQF profits from padding. Let $\query$ be an arbitrary boolean graph property. In the following we construct a ternary $\DynQF$ program $\prog$ which maintains $\query$ if $2^{|\domain^+|^2} = |\domain^-|$. The idea is to identify $\domain^-$ with the set of all graphs over $\domain^+$, that is $\domain^-$ contains an \mbox{element $c_G$} for every graph $G$ over $\domain^+$. A unary relation $R_\query$ stores those elements of $\domain^-$ that correspond to graphs with the property $\query$. Finally the program maintains a pointer $p$ to the element in $\domain^-$ corresponding to the current graph over $\domain^+$. The pointer is updated upon edge modification by using ternary functions $f_\ins$ and $f_\del$ initialized by the initialization mapping in a suitable way. 

  The program $\prog$ is over schema $\schema = \{Q, p, f_\ins, f_\del, R_\query\}$ where $p$ is a constant, $f_\ins$ and $f_\del$ are ternary function symbols, $R_\query$ is a unary relation symbol and $Q$ is the designated query symbol.
  
  We present the initialization mapping of $\prog$  first. The initial state $\state$ for a graph $H$ is defined as follows. The functions $f_\ins$ and $f_\del$ are independent of $H$ and defined via
  \begin{align*}
    f^\state_\ins(a,b,c_G) &= c_{G+(a,b)} \\
    f^\state_\del(a,b,c_G) &= c_{G-(a,b)}
  \end{align*}
  for $a, b \in \domain^+$ and $c_G \in \domain^-$. For all other arguments the value of the functions is arbitrary. Here $G+(a,b)$ and $G-(a,b)$ denote the graphs obtained by adding the edge $(a,b)$ to $G$ and removing the edge $(a,b)$ from $G$, respectively. The relation $R_\query^\state$ contains all $c_G$ with $G \in \query$. Finally the constant $p^\state$ points to $c_H$. 
    
  It remains to exhibit the update formulas. After a modification, the \mbox{pointer $p$} is moved to the node corresponding to the modified graph, and the query bit is updated accordingly:
    \begin{align*}
      \ut{p}{\ins}{u,v}{} &= f_\ins(u,v,p) \quad \quad \quad &\ut{Q}{\ins}{u,v}{} &= R_\query(f_\ins(u,v,p)) \\
      \ut{p}{\del}{u,v}{} &= f_\del(u,v,p) &\ut{Q}{\del}{u,v}{} &= R_\query(f_\del(u,v,p))
    \end{align*}

  Now we sketch how to modify this construction for binary $\DynQF$. The binary $\DynQF$ program maintains $\query$ on an extended non-modifiable domain that contains 
    \begin{itemize}
      \item an element $c_G$ for every graph $G$ over $\domain^+$, and
      \item elements $c_{G,a, \ins}$ and $c_{G,a, \del}$ for every graph $G$ over $\domain^+$ and every $a \in \domain^+$.
    \end{itemize}
  The intuition is that when an edge $(a,b)$ is inserted into the graph $G$ then the pointer $p$ is moved from $c_G$ to the element $c_{G+(a,b)}$ using the intermediate element $c_{G,a, \ins}$. 

  For insertion modifications the binary $\DynQF$ program maintaining $\query$ uses two binary functions $f_\ins$ and $s_\ins$ that are initialized as
  \begin{align*}
    f^\state_\ins(a,c_G) &= c_{G,a, \ins}\\
    s^\state_\ins(b,c_{G,a,\ins}) &= c_{G+(a,b)}
  \end{align*}
  for $a, b \in \domain^+$ and $c_G, c_{G,a, \ins} \in \domain^-$. For all other arguments the value of the functions is arbitrary. 

  When an insertion occurs, the pointer and the query bit are updated via 
  \begin{align*}
    \ut{p}{\ins}{u,v}{} &= s_\ins(v, f_\ins(u,p))\\
    \ut{Q}{\ins}{u,v}{} &= R_\query(s_\ins(v, f_\ins(u,p)))
  \end{align*}
  
  The update formulas and terms for deletions are analogous.   

\end{proof}

Hence the ability to profit from padding distinguishes binary $\DynQF$ and $\DynProp$ extended by unary functions. Although the proof of the preceding theorem requires the non-modifiable domain to be of exponential size with respect to the modifiable domain, the construction also explains why the lower bound technique from the previous sections cannot be immediately applied to binary $\DynQF$. In the lower bound construction only tuples over the set $A$ are modified, while tuples containing elements from $C = \allsubsets{a}{k}$ are not modified. Thus, by treating  $C$ as a non-modifiable domain, it can be used to store information as in the proof above. As the modification sequences used in the lower bounds are of length $k^2$, finding similar substructures in structures with binary auxiliary functions becomes much harder.

  \section{Conclusion and Future Work}
    In this work we exhibited a precise dynamic descriptive complexity characterization of the $k$-clique query when only insertions are allowed. The characterization implies an arity hierarchy for graph queries for \DynProp under insertions. Further we exhibited a very simple $\QFO[\exists^*\forall^*]$-property which is not maintainable in $\DynProp$. We also discussed the limit of our proof methods.

While proving lower bounds for full $\DynFO$ --- a major long-term goals in dynamic descriptive complexity --- might be really hard to achieve, we believe that the following goals are suitable for both developing new lower bound methods and for further improving the current methods.

\begin{goal}
  Prove general quantifier-free lower bounds for insertions and deletions for the reachability query and the $k$-clique query.
\end{goal}

It is known that both queries cannot be maintained in  binary \mbox{$\DynProp$ \cite{ZeumeS15reach}}. We conjecture that neither the $3$-clique query nor the reachability query can be maintained in $\DynProp$ under deletions.

\begin{goal}
  Find a general framework for proving quantifier-free lower bounds.\end{goal}

\begin{goal}
  Find a   query that cannot be maintained in binary $\DynQF$.
\end{goal}

In this and previous work on lower bounds in the dynamic descriptive complexity setting, the theorems of Ramsey and Higman played an important role for lower bound proofs. Therefore it appears to be promissing to study the applicability of other combinatorical tools in this context.

 \section*{References}   
 \bibliographystyle{plain}
 \bibliography{bibliography}
\end{document}